%% file: main.tex
\documentclass[runningheads]{llncs}

\usepackage{head}

\input{headeforpics.tex}

\renewcommand{\orcidID}[1]{\href{https://orcid.org/#1}{\hspace{2px}\includegraphics[width=10px,height=10px]{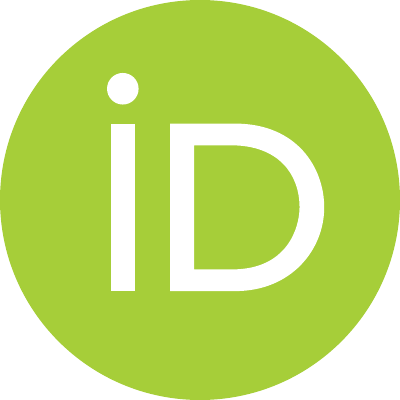}}\hspace{2px}}
\begin{document}
\title{Why quantum state verification cannot be both efficient and secure: a categorical approach}

%
\author{Fabian Wiesner \inst{1}\orcidID{0000-0003-1309-9786}
\and Ziad Chaoui \inst{1} \and
Diana Kessler \inst{2}\orcidID{0009-0006-6355-1581} \and\\ 
Anna Pappa \inst{1}\orcidID{0000-0002-4662-149X} \and
Martti Karvonen\inst{3}\orcidID{0000-0002-8919-343X}}
\authorrunning{F. Wiesner et al.}
\titlerunning{Why quantum state verification cannot be both efficient and secure}

\institute{Technische Universität Berlin \email{\{f.wiesner,ziad.chaoui,anna.pappa\}@tu-berlin.de} \\
 \and
Tallinn University of Technology
\email{diana-maria.kessler@taltech.ee}\and
University College London
\email{martti.karvonen@ucl.ac.uk}}

\maketitle              
\begin{abstract}
The advantage of quantum protocols lies in the inherent properties of the shared quantum states. These states are sometimes provided by sources that are not trusted, and therefore need to be verified. Finding secure and efficient quantum state verification protocols remains a big challenge, and recent works illustrate trade-offs between efficiency and security for different groups of states in restricted settings. However, whether a universal trade-off exists for all quantum states and all verification strategies remains unknown. In this work, we instantiate the categorical composable cryptography framework to show a fundamental limit for quantum state verification for all cut-and-choose approaches used to verify arbitrary quantum states. Our findings show that the prevailing cut-and-choose techniques cannot lead to quantum state verification protocols that are both efficient and secure.

\end{abstract}

\input{subtexs/01_Introduction}
\input{subtexs/02_Methods}

\input{subtexs/03_CCC}
\input{subtexs/04_Results}

\input{subtexs/05_Discussion}

\bibliography{Bibliography}
\appendix
\input{subtexs/a01_Appendix}
\end{document}

%% file: headeforpics.tex
\usepackage{tikz}
\usetikzlibrary{decorations.pathreplacing}
\usetikzlibrary{decorations.markings}
\usetikzlibrary{calc}
\usetikzlibrary{arrows}
\usetikzlibrary{arrows.meta}
\usetikzlibrary{shapes.geometric}
\tikzset{>={To[length=2.5pt,width=4pt]}}

\usetikzlibrary{intersections}

\newenvironment{pic}[1][]
{\begin{aligned}\begin{tikzpicture}[font=\tiny,#1]}
{\end{tikzpicture}\end{aligned}}

\pgfdeclarelayer{foreground}
\pgfdeclarelayer{background}
\pgfdeclarelayer{morphismlayer}
\pgfdeclarelayer{dashedmorphismlayer}
\pgfdeclarelayer{edgelayer}
\pgfdeclarelayer{nodelayer}
\pgfsetlayers{background,main,morphismlayer,dashedmorphismlayer,foreground,edgelayer,nodelayer}

\def\thickness{0.7pt}
\makeatletter
\pgfkeys{%
  /tikz/on layer/.code={
    \pgfonlayer{#1}\begingroup
    \aftergroup\endpgfonlayer
    \aftergroup\endgroup
  },
  /tikz/node on layer/.code={
    \gdef\node@@on@layer{%
      \setbox\tikz@tempbox=\hbox\bgroup\pgfonlayer{#1}\unhbox\tikz@tempbox\endpgfonlayer\pgfsetlinewidth{\thickness}\egroup}
    \aftergroup\node@on@layer
  },
  /tikz/end node on layer/.code={
    \endpgfonlayer\endgroup\endgroup
  }
}
\def\node@on@layer{\aftergroup\node@@on@layer}
\makeatother

\tikzstyle{braid}=[double=black, line width=3*\thickness, double distance=\thickness, draw=white, text=black]
\tikzset{every picture/.style={line width=\thickness, draw=black}}
\tikzstyle{pure}=[line width=.7pt]
\tikzstyle{string}=[line width=\thickness]
\tikzstyle{scalar}=[circle, inner sep=0pt, minimum width=15pt, draw, line width=\thickness, fill=white]
\tikzstyle{dot}=[circle, draw=black, fill=black!25, inner sep=.4ex, line width=\thickness, node on layer=foreground]
\tikzstyle{blackdot}=[circle, draw=black, fill=black!75, inner sep=.4ex, line width=\thickness, node on layer=foreground, text=white]
\tikzstyle{whitedot}=[circle, draw=black, fill=white, inner sep=.4ex, line width=\thickness, node on layer=foreground]
\tikzstyle{mixedmorphism}=[morphism, minimum width=30pt, minimum height=16pt, draw, font=\small, inner sep=0pt, fill=white, line width=\thickness,rounded corners=1ex]

\tikzstyle{triangle} = [regular polygon, regular polygon sides=3, draw=black, fill=black!20,scale=0.4, node on layer=foreground]

\tikzstyle{whitetriangle}=[triangle, fill=white]
\tikzstyle{greytriangle}=[triangle, fill=black!20]
\tikzstyle{darkgreytriangle}=[triangle, fill=black!50]
\tikzstyle{blacktriangle}=[triangle, fill=black]

\tikzstyle{invertedtriangle} = [triangle,scale=-1]
\tikzstyle{whiteinvertedtriangle}=[invertedtriangle, fill=white]
\tikzstyle{greyinvertedtriangle}=[invertedtriangle, fill=black!20]
\tikzstyle{darkgreyinvertedtriangle}=[invertedtriangle, fill=black!50]
\tikzstyle{blackinvertedtriangle}=[invertedtriangle, fill=black]

\tikzset{functor1/.style={gray!50}}
\tikzset{functor2/.style={gray!50}}
\tikzstyle{thick}=[line width=\thickness]
\tikzstyle{tiny}=[font=\tiny]
\tikzset{circlelabel/.style={draw, thick, circle, inner sep=-5pt,
 fill=white, minimum width=14pt, fill opacity=1, node on layer=foreground, font=\scriptsize}}
\tikzset{shade 1 transparent/.style={fill=black!20, fill opacity=0.8}}
\tikzset{shade 2 transparent/.style={fill=black!35, fill opacity=0.8}}
\tikzset{shade 3 transparent/.style={fill=black!50, fill opacity=0.8}}
\tikzset{shade 4 transparent/.style={fill=black!65, fill opacity=0.8}}
\tikzset{shade 1/.style={fill=black!16, fill opacity=1}}
\tikzset{shade 2/.style={fill=black!28, fill opacity=1}}
\tikzset{shade 3/.style={fill=black!40, fill opacity=1}}
\tikzset{shade 4/.style={fill=black!52, fill opacity=1}}

\def\strarr{line width=0.7pt, length=4pt, width=5pt, color=black}

\tikzset{arrow/.style={decoration={
    markings,
    mark=at position #1 with \arrow{>[\strarr]}},
    postaction=decorate},
    reverse arrow/.style={decoration={
    markings,
    mark=at position #1 with {{\arrow{<[\strarr]}}}},
    postaction=decorate}
}

\tikzset{overbrace/.style={
     decoration={brace},
     decorate}
}
\tikzset{underbrace/.style={
     decoration={brace, mirror},
     decorate}
}

\newif\ifblack\pgfkeys{/tikz/black/.is if=black}
\newif\ifwedge\pgfkeys{/tikz/wedge/.is if=wedge}
\newif\ifvflip\pgfkeys{/tikz/vflip/.is if=vflip}
\newif\ifhflip\pgfkeys{/tikz/hflip/.is if=hflip}
\newif\ifhvflip\pgfkeys{/tikz/hvflip/.is if=hvflip}
\newif\ifconnectsw\pgfkeys{/tikz/connect sw/.is if=connectsw}
\newif\ifconnectse\pgfkeys{/tikz/connect se/.is if=connectse}
\newif\ifconnectn\pgfkeys{/tikz/connect n/.is if=connectn}
\newif\ifconnects\pgfkeys{/tikz/connect s/.is if=connects}
\newif\ifconnectnw\pgfkeys{/tikz/connect nw/.is if=connectnw}
\newif\ifconnectne\pgfkeys{/tikz/connect ne/.is if=connectne}
\newif\ifconnectnwf\pgfkeys{/tikz/connect nw >/.is if=connectnwf}
\newif\ifconnectnef\pgfkeys{/tikz/connect ne >/.is if=connectnef}
\newif\ifconnectswf\pgfkeys{/tikz/connect sw >/.is if=connectswf}
\newif\ifconnectsef\pgfkeys{/tikz/connect se >/.is if=connectsef}
\newif\ifconnectnf\pgfkeys{/tikz/connect n >/.is if=connectnf}
\newif\ifconnectsf\pgfkeys{/tikz/connect s >/.is if=connectsf}
\newif\ifconnectnwr\pgfkeys{/tikz/connect nw </.is if=connectnwr}
\newif\ifconnectner\pgfkeys{/tikz/connect ne </.is if=connectner}
\newif\ifconnectswr\pgfkeys{/tikz/connect sw </.is if=connectswr}
\newif\ifconnectser\pgfkeys{/tikz/connect se </.is if=connectser}
\newif\ifconnectnr\pgfkeys{/tikz/connect n </.is if=connectnr}
\newif\ifconnectsr\pgfkeys{/tikz/connect s </.is if=connectsr}
\tikzset{keylengthnw/.initial=\connectheight}
\tikzset{keylengthn/.initial =\connectheight}
\tikzset{keylengthne/.initial=\connectheight}
\tikzset{keylengthsw/.initial=\connectheight}
\tikzset{keylengths/.initial =\connectheight}
\tikzset{keylengthse/.initial=\connectheight}
\tikzset{connect nw length/.style={connect nw=true, keylengthnw={#1}}}
\tikzset{connect n length/.style ={connect n =true, keylengthn ={#1}}}
\tikzset{connect ne length/.style={connect ne=true, keylengthne={#1}}}
\tikzset{connect sw length/.style={connect sw=true, keylengthsw={#1}}}
\tikzset{connect s length/.style ={connect s =true, keylengths ={#1}}}
\tikzset{connect se length/.style={connect se=true, keylengthse={#1}}}
\tikzset{connect nw < length/.style={connect nw <=true, keylengthnw={#1}}}
\tikzset{connect n < length/.style ={connect n <=true,  keylengthn ={#1}}}
\tikzset{connect ne < length/.style={connect ne <=true, keylengthne={#1}}}
\tikzset{connect sw < length/.style={connect sw <=true, keylengthnw={#1}}}
\tikzset{connect s < length/.style ={connect s <=true,  keylengths ={#1}}}
\tikzset{connect se < length/.style={connect se <=true, keylengthse={#1}}}
\tikzset{connect nw > length/.style={connect nw >=true, keylengthnw={#1}}}
\tikzset{connect n > length/.style ={connect n >=true,  keylengthn ={#1}}}
\tikzset{connect ne > length/.style={connect ne >=true, keylengthne={#1}}}
\tikzset{connect sw > length/.style={connect sw >=true, keylengthsw={#1}}}
\tikzset{connect s > length/.style ={connect s >=true,  keylengths ={#1}}}
\tikzset{connect se > length/.style={connect se >=true, keylengthse={#1}}}

\newlength\morphismheight
\newlength\wedgewidth
\newlength\minimummorphismwidth
\newlength\stateheight
\newlength\minimumstatewidth
\newlength\connectheight
\tikzset{width/.initial=\minimummorphismwidth}
\tikzset{colour/.initial=white}

\makeatletter
\pgfarrowsdeclare{thickarrow}{thickarrow}
{
  \pgfutil@tempdima=-0.84pt%
  \advance\pgfutil@tempdima by-1.3\pgflinewidth%
  \pgfutil@tempdimb=-1.7pt%
  \advance\pgfutil@tempdimb by.625\pgflinewidth%
  \pgfarrowsleftextend{+\pgfutil@tempdima}
  \pgfarrowsrightextend{+\pgfutil@tempdimb}
}
{
  \pgfmathparse{\pgfgetarrowoptions{thickarrow}}%
  \pgfsetlinewidth{1.25 pt}
  \pgfutil@tempdima=0.28pt%
  \advance\pgfutil@tempdima by.3\pgflinewidth%
  \pgfsetlinewidth{0.8\pgflinewidth}
  \pgfsetdash{}{+0pt}
  \pgfsetroundcap
  \pgfsetroundjoin
  \pgfpathmoveto{\pgfqpoint{-3\pgfutil@tempdima}{4\pgfutil@tempdima}}
  \pgfpathcurveto
  {\pgfqpoint{-2.75\pgfutil@tempdima}{2.5\pgfutil@tempdima}}
  {\pgfqpoint{0pt}{0.25\pgfutil@tempdima}}
  {\pgfqpoint{0.75\pgfutil@tempdima}{0pt}}
  \pgfpathcurveto
  {\pgfqpoint{0pt}{-0.25\pgfutil@tempdima}}
  {\pgfqpoint{-2.75\pgfutil@tempdima}{-2.5\pgfutil@tempdima}}
  {\pgfqpoint{-3\pgfutil@tempdima}{-4\pgfutil@tempdima}}
  \pgfusepathqstroke
}
\pgfarrowsdeclare{reversethickarrow}{reversethickarrow}
{
  \pgfutil@tempdima=-0.84pt%
  \advance\pgfutil@tempdima by-1.3\pgflinewidth%
  \pgfutil@tempdimb=0.2pt%
  \advance\pgfutil@tempdimb by.625\pgflinewidth%
  \pgfarrowsleftextend{+\pgfutil@tempdima}
  \pgfarrowsrightextend{+\pgfutil@tempdimb}
}
{
  \pgftransformxscale{-1}
  \pgfmathparse{\pgfgetarrowoptions{thickarrow}}%
  \ifpgfmathunitsdeclared%
    \pgfmathparse{\pgfmathresult pt}%
  \else%
    \pgfmathparse{\pgfmathresult*\pgflinewidth}%
  \fi%
  \let\thickness=\pgfmathresult
  \pgfsetlinewidth{1.25 pt}
  \pgfutil@tempdima=0.28pt%
  \advance\pgfutil@tempdima by.3\pgflinewidth%
  \pgfsetlinewidth{0.8\pgflinewidth}
  \pgfsetdash{}{+0pt}
  \pgfsetroundcap
  \pgfsetroundjoin
  \pgfpathmoveto{\pgfqpoint{-3\pgfutil@tempdima}{4\pgfutil@tempdima}}
  \pgfpathcurveto
  {\pgfqpoint{-2.75\pgfutil@tempdima}{2.5\pgfutil@tempdima}}
  {\pgfqpoint{0pt}{0.25\pgfutil@tempdima}}
  {\pgfqpoint{0.75\pgfutil@tempdima}{0pt}}
  \pgfpathcurveto
  {\pgfqpoint{0pt}{-0.25\pgfutil@tempdima}}
  {\pgfqpoint{-2.75\pgfutil@tempdima}{-2.5\pgfutil@tempdima}}
  {\pgfqpoint{-3\pgfutil@tempdima}{-4\pgfutil@tempdima}}
  \pgfusepathqstroke
}
\makeatother

\tikzset{diredge/.style={decoration={
  markings,
  mark=at position 0.525 with {\arrow{#1}}},postaction={decorate}}}
\tikzset{
    diredge/.default=>
}

\tikzset{diredgestart/.style={decoration={
  markings,
  mark=at position 4pt with {\arrow{#1}}},postaction={decorate}}}
\tikzset{
    diredgestart/.default=<
}

\tikzset{diredgeend/.style={decoration={
  markings,
  mark=at position 1 with {\arrow{#1}}},postaction={decorate}}}
\tikzset{
    diredgeend/.default=>
}

\makeatletter
\pgfdeclareshape{ground}
{
    \savedanchor\centerpoint
    {
        \pgf@x=0pt
        \pgf@y=0pt
    }
    \anchor{center}{\centerpoint}
    \anchorborder{\centerpoint}
    \anchor{north}
    {
        \pgf@x=0pt
        \pgf@y=0.5*0.33*\stateheight
    }
    \anchor{south}
    {
        \pgf@x=0pt
        \pgf@y=0pt
    }
    \saveddimen\overallwidth
    {
        \pgfkeysgetvalue{/pgf/minimum width}{\minwidth}
        \pgf@x=\minimumstatewidth
        \ifdim\pgf@x<\minwidth
            \pgf@x=\minwidth
        \fi
    }
    \backgroundpath
    {
        \begin{pgfonlayer}{foreground}
        \pgfsetstrokecolor{black}
        \pgfsetlinewidth{1.25pt}
        \ifhflip
            \pgftransformyscale{-1}
        \fi
        \pgftransformscale{0.5}
        \pgfpathmoveto{\pgfpoint{-0.5*\overallwidth}{0}}
        \pgfpathlineto{\pgfpoint{0.5*\overallwidth}{0}}
        \pgfpathmoveto{\pgfpoint{-0.33*\overallwidth}{0.33*\stateheight}}
        \pgfpathlineto{\pgfpoint{0.33*\overallwidth}{0.33*\stateheight}}
        \pgfpathmoveto{\pgfpoint{-0.16*\overallwidth}{0.66*\stateheight}}
        \pgfpathlineto{\pgfpoint{0.16*\overallwidth}{0.66*\stateheight}}
        \pgfpathmoveto{\pgfpoint{-0.02*\overallwidth}{\stateheight}}
        \pgfpathlineto{\pgfpoint{0.02*\overallwidth}{\stateheight}}
        \pgfusepath{stroke}
        \end{pgfonlayer}
    }
}
\tikzset{forward arrow style/.style={every to/.style, decoration={
    markings,
    mark=at position 0.5*\pgfdecoratedpathlength+2pt with \arrow{>[\strarr]}},
    postaction=decorate}}
\tikzset{reverse arrow style/.style={every to/.style, decoration={
    markings,
    mark=at position 0.5*\pgfdecoratedpathlength+2pt with \arrow{<[\strarr]}},
    postaction=decorate}}
\pgfdeclareshape{morphismshape}
{
    \savedanchor\centerpoint
    {
        \pgf@x=0pt
        \pgf@y=0pt
    }
    \anchor{center}{\centerpoint}
    \anchorborder{\centerpoint}
    \saveddimen\savedlengthnw
    {
        \pgfkeysgetvalue{/tikz/keylengthnw}{\len}
        \pgf@x=\len
    }
    \saveddimen\savedlengthn
    {
        \pgfkeysgetvalue{/tikz/keylengthn}{\len}
        \pgf@x=\len
    }
    \saveddimen\savedlengthne
    {
        \pgfkeysgetvalue{/tikz/keylengthne}{\len}
        \pgf@x=\len
    }
    \saveddimen\savedlengthsw
    {
        \pgfkeysgetvalue{/tikz/keylengthsw}{\len}
        \pgf@x=\len
    }
    \saveddimen\savedlengths
    {
        \pgfkeysgetvalue{/tikz/keylengths}{\len}
        \pgf@x=\len
    }
    \saveddimen\savedlengthse
    {
        \pgfkeysgetvalue{/tikz/keylengthse}{\len}
        \pgf@x=\len
    }
    \saveddimen\overallwidth
    {
        \pgfkeysgetvalue{/tikz/width}{\minwidth}
        \pgf@x=\wd\pgfnodeparttextbox
        \ifdim\pgf@x<\minwidth
            \pgf@x=\minwidth
        \fi
    }
    \savedanchor{\upperrightcorner}
    {
        \pgf@y=.5\ht\pgfnodeparttextbox
        \advance\pgf@y by -.5\dp\pgfnodeparttextbox
        \pgf@x=.5\wd\pgfnodeparttextbox
    }
    \anchor{north}
    {
        \pgf@x=0pt
        \pgf@y=0.5\morphismheight
    }
    \anchor{north east}
    {
        \pgf@x=\overallwidth
        \multiply \pgf@x by 2
        \divide \pgf@x by 5
        \pgf@y=0.5\morphismheight
    }
    \anchor{east}
    {
        \pgf@x=\overallwidth
        \divide \pgf@x by 2
        \advance \pgf@x by 5pt
        \pgf@y=0pt
    }
    \anchor{west}
    {
        \pgf@x=-\overallwidth
        \divide \pgf@x by 2
        \advance \pgf@x by -5pt
        \pgf@y=0pt
    }
    \anchor{north west}
    {
        \pgf@x=-\overallwidth
        \multiply \pgf@x by 2
        \divide \pgf@x by 5
        \pgf@y=0.5\morphismheight
    }
    \anchor{connect nw}
    {
        \pgf@x=-\overallwidth
        \multiply \pgf@x by 2
        \divide \pgf@x by 5
        \pgf@y=0.5\morphismheight
        \advance\pgf@y by \savedlengthnw
    }
    \anchor{connect ne}
    {
        \pgf@x=\overallwidth
        \multiply \pgf@x by 2
        \divide \pgf@x by 5
        \pgf@y=0.5\morphismheight
        \advance\pgf@y by \savedlengthne
    }
    \anchor{connect sw}
    {
        \pgf@x=-\overallwidth
        \multiply \pgf@x by 2
        \divide \pgf@x by 5
        \pgf@y=-0.5\morphismheight
        \advance\pgf@y by -\savedlengthsw
    }
    \anchor{connect se}
    {
        \pgf@x=\overallwidth
        \multiply \pgf@x by 2
        \divide \pgf@x by 5
        \pgf@y=-0.5\morphismheight
        \advance\pgf@y by -\savedlengthse
    }
    \anchor{connect n}
    {
        \pgf@x=0pt
        \pgf@y=0.5\morphismheight
        \advance\pgf@y by \savedlengthn
    }
    \anchor{connect s}
    {
        \pgf@x=0pt
        \pgf@y=-0.5\morphismheight
        \advance\pgf@y by -\savedlengths
    }
    \anchor{south east}
    {
        \pgf@x=\overallwidth
        \multiply \pgf@x by 2
        \divide \pgf@x by 5
        \pgf@y=-0.5\morphismheight
    }
    \anchor{south west}
    {
        \pgf@x=-\overallwidth
        \multiply \pgf@x by 2
        \divide \pgf@x by 5
        \pgf@y=-0.5\morphismheight
    }
    \anchor{south}
    {
        \pgf@x=0pt
        \pgf@y=-0.5\morphismheight
    }
    \anchor{text}
    {
        \upperrightcorner
        \pgf@x=-\pgf@x
        \pgf@y=-\pgf@y
    }
    \backgroundpath
    {
    \begin{scope}
        \pgfkeysgetvalue{/tikz/fill}{\morphismfill}
        \pgfsetstrokecolor{black}
        \pgfsetlinewidth{\thickness}
        \begin{scope}
            \begin{pgfonlayer}{morphismlayer}
        \pgfsetstrokecolor{black}
        \pgfsetfillcolor{\pgfkeysvalueof{/tikz/colour}}
        \pgfsetlinewidth{\thickness}
                \ifhflip
                    \pgftransformyscale{-1}
                \fi
                \ifvflip
                    \pgftransformxscale{-1}
                \fi
                \ifhvflip
                    \pgftransformxscale{-1}
                    \pgftransformyscale{-1}
                \fi
                \pgfpathmoveto{\pgfpoint
                    {-0.5*\overallwidth-5pt}
                    {0.5*\morphismheight}}
                \pgfpathlineto{\pgfpoint
                    {0.5*\overallwidth+5pt}
                    {0.5*\morphismheight}}
                \ifwedge
                    \pgfpathlineto{\pgfpoint
                        {0.5*\overallwidth + \wedgewidth}
                        {-0.5*\morphismheight}}
                \else
                    \pgfpathlineto{\pgfpoint
                        {0.5*\overallwidth + 5pt}
                        {-0.5*\morphismheight}}
                \fi
                \pgfpathlineto{\pgfpoint
                    {-0.5*\overallwidth-5pt}
                    {-0.5*\morphismheight}}
                \pgfpathclose
                \pgfusepath{fill,stroke}
            \end{pgfonlayer}
        \end{scope}
        \ifconnectnw
            \pgfpathmoveto{\pgfpoint
                {-0.4*\overallwidth}
                {0.5*\morphismheight}}
            \pgfpathlineto{\pgfpoint
                {-0.4*\overallwidth}
                {0.5*\morphismheight+\savedlengthnw}}
            \pgfusepath{stroke}
        \fi
        \ifconnectne
            \pgfpathmoveto{\pgfpoint
                {0.4*\overallwidth}
                {0.5*\morphismheight}}
            \pgfpathlineto{\pgfpoint
                {0.4*\overallwidth}
                {0.5*\morphismheight+\savedlengthne}}
            \pgfusepath{stroke}
        \fi
        \ifconnectsw
            \pgfpathmoveto{\pgfpoint
                {-0.4*\overallwidth}
                {-0.5*\morphismheight}}
            \pgfpathlineto{\pgfpoint
                {-0.4*\overallwidth}
                {-0.5*\morphismheight-\savedlengthsw}}
            \pgfusepath{stroke}
        \fi
        \ifconnectse
            \pgfpathmoveto{\pgfpoint
                {0.4*\overallwidth}
                {-0.5*\morphismheight}}
            \pgfpathlineto{\pgfpoint
                {0.4*\overallwidth}
                {-0.5*\morphismheight-\savedlengthse}}
            \pgfusepath{stroke}
        \fi
        \ifconnectn
            \pgfpathmoveto{\pgfpoint
                {0pt}
                {0.5*\morphismheight}}
            \pgfpathlineto{\pgfpoint
                {0pt}
                {0.5*\morphismheight+\savedlengthn}}
            \pgfusepath{stroke}
        \fi
        \ifconnects
            \pgfpathmoveto{\pgfpoint
                {0pt}
                {-0.5*\morphismheight}}
            \pgfpathlineto{\pgfpoint
                {0pt}
                {-0.5*\morphismheight-\savedlengths}}
            \pgfusepath{stroke}
        \fi
        \ifconnectnwf
            \draw [forward arrow style] (-0.4*\overallwidth,0.5*\morphismheight)
                to (-0.4*\overallwidth,0.5*\morphismheight+\savedlengthnw);
        \fi
        \ifconnectnef
            \draw [forward arrow style] (0.4*\overallwidth,0.5*\morphismheight)
                to (0.4*\overallwidth,0.5*\morphismheight+\savedlengthne);
        \fi
        \ifconnectswf
            \draw [forward arrow style] (-0.4*\overallwidth,-0.5*\morphismheight-\savedlengthsw)
                to (-0.4*\overallwidth,-0.5*\morphismheight);
        \fi
        \ifconnectsef
            \draw [forward arrow style] (0.4*\overallwidth,-0.5*\morphismheight-\savedlengthse)
                to (0.4*\overallwidth,-0.5*\morphismheight);
        \fi
        \ifconnectnf
            \draw [forward arrow style] (0,0.5*\morphismheight)
                to (0,0.5*\morphismheight+\savedlengthn);
        \fi
        \ifconnectsf
            \draw [forward arrow style] (0,-0.5*\morphismheight-\savedlengths)
                to (0,-0.5*\morphismheight);
        \fi
        \ifconnectnwr
            \draw [reverse arrow style] (-0.4*\overallwidth,0.5*\morphismheight)
                to (-0.4*\overallwidth,0.5*\morphismheight+\savedlengthnw);
        \fi
        \ifconnectner
            \draw [reverse arrow style] (0.4*\overallwidth,0.5*\morphismheight)
                to (0.4*\overallwidth,0.5*\morphismheight+\savedlengthne);
        \fi
        \ifconnectswr
            \draw [reverse arrow style] (-0.4*\overallwidth,-0.5*\morphismheight-\savedlengthsw)
                to (-0.4*\overallwidth,-0.5*\morphismheight);
        \fi
        \ifconnectser
            \draw [reverse arrow style] (0.4*\overallwidth,-0.5*\morphismheight-\savedlengthse)
                to (0.4*\overallwidth,-0.5*\morphismheight);
        \fi
        \ifconnectnr
            \draw [reverse arrow style] (0,0.5*\morphismheight)
                to (0,0.5*\morphismheight+\savedlengthn);
        \fi
        \ifconnectsr
            \draw [reverse arrow style] (0,-0.5*\morphismheight-\savedlengths)
                to (0,-0.5*\morphismheight);
        \fi
    \end{scope}
    }
}
\tikzset{morphism/.style={morphismshape, node on layer=foreground}}
\pgfdeclareshape{dashedmorphismshape}
{
    \savedanchor\centerpoint
    {
        \pgf@x=0pt
        \pgf@y=0pt
    }
    \anchor{center}{\centerpoint}
    \anchorborder{\centerpoint}
    \saveddimen\savedlengthnw
    {
        \pgfkeysgetvalue{/tikz/keylengthnw}{\len}
        \pgf@x=\len
    }
    \saveddimen\savedlengthn
    {
        \pgfkeysgetvalue{/tikz/keylengthn}{\len}
        \pgf@x=\len
    }
    \saveddimen\savedlengthne
    {
        \pgfkeysgetvalue{/tikz/keylengthne}{\len}
        \pgf@x=\len
    }
    \saveddimen\savedlengthsw
    {
        \pgfkeysgetvalue{/tikz/keylengthsw}{\len}
        \pgf@x=\len
    }
    \saveddimen\savedlengths
    {
        \pgfkeysgetvalue{/tikz/keylengths}{\len}
        \pgf@x=\len
    }
    \saveddimen\savedlengthse
    {
        \pgfkeysgetvalue{/tikz/keylengthse}{\len}
        \pgf@x=\len
    }
    \saveddimen\overallwidth
    {
        \pgfkeysgetvalue{/tikz/width}{\minwidth}
        \pgf@x=\wd\pgfnodeparttextbox
        \ifdim\pgf@x<\minwidth
            \pgf@x=\minwidth
        \fi
    }
    \savedanchor{\upperrightcorner}
    {
        \pgf@y=.5\ht\pgfnodeparttextbox
        \advance\pgf@y by -.5\dp\pgfnodeparttextbox
        \pgf@x=.5\wd\pgfnodeparttextbox
    }
    \anchor{north}
    {
        \pgf@x=0pt
        \pgf@y=0.5\morphismheight
    }
    \anchor{north east}
    {
        \pgf@x=\overallwidth
        \multiply \pgf@x by 2
        \divide \pgf@x by 5
        \pgf@y=0.5\morphismheight
    }
    \anchor{east}
    {
        \pgf@x=\overallwidth
        \divide \pgf@x by 2
        \advance \pgf@x by 5pt
        \pgf@y=0pt
    }
    \anchor{west}
    {
        \pgf@x=-\overallwidth
        \divide \pgf@x by 2
        \advance \pgf@x by -5pt
        \pgf@y=0pt
    }
    \anchor{north west}
    {
        \pgf@x=-\overallwidth
        \multiply \pgf@x by 2
        \divide \pgf@x by 5
        \pgf@y=0.5\morphismheight
    }
    \anchor{connect nw}
    {
        \pgf@x=-\overallwidth
        \multiply \pgf@x by 2
        \divide \pgf@x by 5
        \pgf@y=0.5\morphismheight
        \advance\pgf@y by \savedlengthnw
    }
    \anchor{connect ne}
    {
        \pgf@x=\overallwidth
        \multiply \pgf@x by 2
        \divide \pgf@x by 5
        \pgf@y=0.5\morphismheight
        \advance\pgf@y by \savedlengthne
    }
    \anchor{connect sw}
    {
        \pgf@x=-\overallwidth
        \multiply \pgf@x by 2
        \divide \pgf@x by 5
        \pgf@y=-0.5\morphismheight
        \advance\pgf@y by -\savedlengthsw
    }
    \anchor{connect se}
    {
        \pgf@x=\overallwidth
        \multiply \pgf@x by 2
        \divide \pgf@x by 5
        \pgf@y=-0.5\morphismheight
        \advance\pgf@y by -\savedlengthse
    }
    \anchor{connect n}
    {
        \pgf@x=0pt
        \pgf@y=0.5\morphismheight
        \advance\pgf@y by \savedlengthn
    }
    \anchor{connect s}
    {
        \pgf@x=0pt
        \pgf@y=-0.5\morphismheight
        \advance\pgf@y by -\savedlengths
    }
    \anchor{south east}
    {
        \pgf@x=\overallwidth
        \multiply \pgf@x by 2
        \divide \pgf@x by 5
        \pgf@y=-0.5\morphismheight
    }
    \anchor{south west}
    {
        \pgf@x=-\overallwidth
        \multiply \pgf@x by 2
        \divide \pgf@x by 5
        \pgf@y=-0.5\morphismheight
    }
    \anchor{south}
    {
        \pgf@x=0pt
        \pgf@y=-0.5\morphismheight
    }
    \anchor{text}
    {
        \upperrightcorner
        \pgf@x=-\pgf@x
        \pgf@y=-\pgf@y
    }
    \backgroundpath
    {
    \begin{scope}
        \pgfkeysgetvalue{/tikz/fill}{\morphismfill}
        \pgfsetstrokecolor{black}
        \pgfsetlinewidth{\thickness}
        \begin{scope}
            \begin{pgfonlayer}{dashedmorphismlayer}
        \pgfsetstrokecolor{black}
        \pgfsetdash{{3pt}{3pt}}{0pt}
        \pgfsetfillcolor{\pgfkeysvalueof{/tikz/colour}}
        \pgfsetlinewidth{\thickness}
                \ifhflip
                    \pgftransformyscale{-1}
                \fi
                \ifvflip
                    \pgftransformxscale{-1}
                \fi
                \ifhvflip
                    \pgftransformxscale{-1}
                    \pgftransformyscale{-1}
                \fi
                \pgfpathmoveto{\pgfpoint
                    {-0.5*\overallwidth-5pt}
                    {0.5*\morphismheight}}
                \pgfpathlineto{\pgfpoint
                    {0.5*\overallwidth+5pt}
                    {0.5*\morphismheight}}
                \ifwedge
                    \pgfpathlineto{\pgfpoint
                        {0.5*\overallwidth + \wedgewidth}
                        {-0.5*\morphismheight}}
                \else
                    \pgfpathlineto{\pgfpoint
                        {0.5*\overallwidth + 5pt}
                        {-0.5*\morphismheight}}
                \fi
                \pgfpathlineto{\pgfpoint
                    {-0.5*\overallwidth-5pt}
                    {-0.5*\morphismheight}}
                \pgfpathclose
                \pgfusepath{fill,stroke}
            \end{pgfonlayer}
        \end{scope}
        \ifconnectnw
            \pgfpathmoveto{\pgfpoint
                {-0.4*\overallwidth}
                {0.5*\morphismheight}}
            \pgfpathlineto{\pgfpoint
                {-0.4*\overallwidth}
                {0.5*\morphismheight+\savedlengthnw}}
            \pgfusepath{stroke}
        \fi
        \ifconnectne
            \pgfpathmoveto{\pgfpoint
                {0.4*\overallwidth}
                {0.5*\morphismheight}}
            \pgfpathlineto{\pgfpoint
                {0.4*\overallwidth}
                {0.5*\morphismheight+\savedlengthne}}
            \pgfusepath{stroke}
        \fi
        \ifconnectsw
            \pgfpathmoveto{\pgfpoint
                {-0.4*\overallwidth}
                {-0.5*\morphismheight}}
            \pgfpathlineto{\pgfpoint
                {-0.4*\overallwidth}
                {-0.5*\morphismheight-\savedlengthsw}}
            \pgfusepath{stroke}
        \fi
        \ifconnectse
            \pgfpathmoveto{\pgfpoint
                {0.4*\overallwidth}
                {-0.5*\morphismheight}}
            \pgfpathlineto{\pgfpoint
                {0.4*\overallwidth}
                {-0.5*\morphismheight-\savedlengthse}}
            \pgfusepath{stroke}
        \fi
        \ifconnectn
            \pgfpathmoveto{\pgfpoint
                {0pt}
                {0.5*\morphismheight}}
            \pgfpathlineto{\pgfpoint
                {0pt}
                {0.5*\morphismheight+\savedlengthn}}
            \pgfusepath{stroke}
        \fi
        \ifconnects
            \pgfpathmoveto{\pgfpoint
                {0pt}
                {-0.5*\morphismheight}}
            \pgfpathlineto{\pgfpoint
                {0pt}
                {-0.5*\morphismheight-\savedlengths}}
            \pgfusepath{stroke}
        \fi
        \ifconnectnwf
            \draw [forward arrow style] (-0.4*\overallwidth,0.5*\morphismheight)
                to (-0.4*\overallwidth,0.5*\morphismheight+\savedlengthnw);
        \fi
        \ifconnectnef
            \draw [forward arrow style] (0.4*\overallwidth,0.5*\morphismheight)
                to (0.4*\overallwidth,0.5*\morphismheight+\savedlengthne);
        \fi
        \ifconnectswf
            \draw [forward arrow style] (-0.4*\overallwidth,-0.5*\morphismheight-\savedlengthsw)
                to (-0.4*\overallwidth,-0.5*\morphismheight);
        \fi
        \ifconnectsef
            \draw [forward arrow style] (0.4*\overallwidth,-0.5*\morphismheight-\savedlengthse)
                to (0.4*\overallwidth,-0.5*\morphismheight);
        \fi
        \ifconnectnf
            \draw [forward arrow style] (0,0.5*\morphismheight)
                to (0,0.5*\morphismheight+\savedlengthn);
        \fi
        \ifconnectsf
            \draw [forward arrow style] (0,-0.5*\morphismheight-\savedlengths)
                to (0,-0.5*\morphismheight);
        \fi
        \ifconnectnwr
            \draw [reverse arrow style] (-0.4*\overallwidth,0.5*\morphismheight)
                to (-0.4*\overallwidth,0.5*\morphismheight+\savedlengthnw);
        \fi
        \ifconnectner
            \draw [reverse arrow style] (0.4*\overallwidth,0.5*\morphismheight)
                to (0.4*\overallwidth,0.5*\morphismheight+\savedlengthne);
        \fi
        \ifconnectswr
            \draw [reverse arrow style] (-0.4*\overallwidth,-0.5*\morphismheight-\savedlengthsw)
                to (-0.4*\overallwidth,-0.5*\morphismheight);
        \fi
        \ifconnectser
            \draw [reverse arrow style] (0.4*\overallwidth,-0.5*\morphismheight-\savedlengthse)
                to (0.4*\overallwidth,-0.5*\morphismheight);
        \fi
        \ifconnectnr
            \draw [reverse arrow style] (0,0.5*\morphismheight)
                to (0,0.5*\morphismheight+\savedlengthn);
        \fi
        \ifconnectsr
            \draw [reverse arrow style] (0,-0.5*\morphismheight-\savedlengths)
                to (0,-0.5*\morphismheight);
        \fi
    \end{scope}
    }
}
\tikzset{dashedmorphism/.style={dashedmorphismshape, node on layer=foreground}}
\pgfdeclareshape{swish right}
{
    \savedanchor\centerpoint
    {
        \pgf@x=0pt
        \pgf@y=0pt
    }
    \anchor{center}{\centerpoint}
    \anchorborder{\centerpoint}
    \anchor{north}
    {
        \pgf@x=\minimummorphismwidth
        \divide\pgf@x by 5
        \pgf@y=\morphismheight
        \divide\pgf@y by 2
        \advance\pgf@y by \connectheight
    }
    \anchor{south}
    {
        \pgf@x=-\minimummorphismwidth
        \divide\pgf@x by 5
        \pgf@y=-\morphismheight
        \divide\pgf@y by 2
        \advance\pgf@y by -\connectheight
    }
    \backgroundpath
    {
        \pgfsetstrokecolor{black}
        \pgfsetlinewidth{\thickness}
        \pgfpathmoveto{\pgfpoint
            {-0.2*\minimummorphismwidth}
            {-0.5*\morphismheight-\connectheight}}
        \pgfpathcurveto
            {\pgfpoint{-0.2*\minimummorphismwidth}{0pt}}
            {\pgfpoint{0.2*\minimummorphismwidth}{0pt}}
            {\pgfpoint
                {0.2*\minimummorphismwidth}
                {0.5*\morphismheight+\connectheight}}
        \pgfusepath{stroke}
    }
}
\pgfdeclareshape{swish left}
{
    \savedanchor\centerpoint
    {
        \pgf@x=0pt
        \pgf@y=0pt
    }
    \anchor{center}{\centerpoint}
    \anchorborder{\centerpoint}
    \anchor{north}
    {
        \pgf@x=-\minimummorphismwidth
        \divide\pgf@x by 5
        \pgf@y=\morphismheight
        \divide\pgf@y by 2
        \advance\pgf@y by \connectheight
    }
    \anchor{south}
    {
        \pgf@x=\minimummorphismwidth
        \divide\pgf@x by 5
        \pgf@y=-\morphismheight
        \divide\pgf@y by 2
        \advance\pgf@y by -\connectheight
    }
    \backgroundpath
    {
        \pgfsetstrokecolor{black}
        \pgfsetlinewidth{\thickness}
        \pgfpathmoveto{\pgfpoint
            {0.2*\minimummorphismwidth}
            {-0.5*\morphismheight-\connectheight}}
        \pgfpathcurveto
            {\pgfpoint{0.2*\minimummorphismwidth}{0pt}}
            {\pgfpoint{-0.2*\minimummorphismwidth}{0pt}}
            {\pgfpoint
                {-0.2*\minimummorphismwidth}
                {0.5*\morphismheight+\connectheight}}
        \pgfusepath{stroke}
    }
}
\pgfdeclareshape{state}
{
    \savedanchor\centerpoint
    {
        \pgf@x=0pt
        \pgf@y=0pt
    }
    \anchor{center}{\centerpoint}
    \anchorborder{\centerpoint}
    \saveddimen\overallwidth
    {
        \pgf@x=3\wd\pgfnodeparttextbox
        \ifdim\pgf@x<\minimumstatewidth
            \pgf@x=\minimumstatewidth
        \fi
    }
    \savedanchor{\upperrightcorner}
    {
        \pgf@x=.5\wd\pgfnodeparttextbox
        \pgf@y=.5\ht\pgfnodeparttextbox
        \advance\pgf@y by -.5\dp\pgfnodeparttextbox
    }
    \anchor{north}
    {
        \pgf@x=0pt
        \pgf@y=0pt
    }
    \anchor{south}
    {
        \pgf@x=0pt
        \pgf@y=0pt
    }
    \anchor{A}
    {
        \pgf@x=-\overallwidth
        \divide\pgf@x by 4
        \pgf@y=0pt
    }
    \anchor{B}
    {
        \pgf@x=\overallwidth
        \divide\pgf@x by 4
        \pgf@y=0pt
    }
    \anchor{text}
    {
        \upperrightcorner
        \pgf@x=-\pgf@x
        \ifhflip
            \pgf@y=-\pgf@y
            \advance\pgf@y by 0.4\stateheight
        \else
            \pgf@y=-\pgf@y
            \advance\pgf@y by -0.4\stateheight
        \fi
    }
    \backgroundpath
    {
       \begin{pgfonlayer}{foreground}
        \pgfsetstrokecolor{black}
        \pgfsetlinewidth{\thickness}
        \pgfpathmoveto{\pgfpoint{-0.5*\overallwidth}{0}}
        \pgfpathlineto{\pgfpoint{0.5*\overallwidth}{0}}
        \ifhflip
            \pgfpathlineto{\pgfpoint{0}{\stateheight}}
        \else
            \pgfpathlineto{\pgfpoint{0}{-\stateheight}}
        \fi
        \pgfpathclose
        \ifblack
            \pgfsetfillcolor{black!50}
            \pgfusepath{fill,stroke}
        \else
            \pgfusepath{stroke}
        \fi
       \end{pgfonlayer}
    }
}
\makeatother

\newcommand{\tinyhandle}[1][dot]{\ensuremath{\smash{\raisebox{-1pt}{\ensuremath{\hspace{-2pt}\begin{pic}[scale=0.33]
        \node (0) at (0,0) {};
        \node[dot, inner sep=1.0pt] (1) at (0,0.3) {};
        \node[dot, inner sep=1.0pt] (2) at (0,0.7) {};
        \node (3) at (0,1) {};
        \draw (0.center) to (1.center);
        \draw (2.center) to (3.center);
        \draw[in=180, out=180, looseness=2] (1.center) to (2.center);
        \draw[in=0, out=0, looseness=2] (1.center) to (2.center);
\end{pic}\hspace{-1pt}}}}}}

\makeatletter
\pgfdeclareshape{arrow shape}
{
    \savedanchor\centerpoint
    {
        \pgf@x=0pt
        \pgf@y=0pt
    }
    \anchor{center}{\centerpoint}
    \anchor{south}{\centerpoint}
    \anchor{north}{\centerpoint}
    \anchorborder{\centerpoint}
    \backgroundpath
    {
        \begin{pgfonlayer}{foreground}
        \pgfsetstrokecolor{black}
        \pgfsetarrowsstart{>[\strarr]}
        \pgfsetlinewidth{\thickness}
        \pgfpathmoveto{\pgfpoint{0}{-2pt}}
        \pgfpathlineto{\pgfpoint{0}{1pt}}
        \pgfusepath{stroke}
        \end{pgfonlayer}
    }
}
\pgfdeclareshape{double arrow shape}
{
    \savedanchor\centerpoint
    {
        \pgf@x=0pt
        \pgf@y=0pt
    }
    \anchor{center}{\centerpoint}
    \anchor{south}{\centerpoint}
    \anchor{north}{\centerpoint}
    \anchorborder{\centerpoint}
    \backgroundpath
    {
        \begin{pgfonlayer}{foreground}
        \pgfsetstrokecolor{black}
        \pgfsetarrowsstart{>[\strarr]}
        \pgfsetlinewidth{\thickness}
        \pgfpathmoveto{\pgfpoint{0}{-3.5pt}}
        \pgfpathlineto{\pgfpoint{0}{-0.5pt}}
        \pgfusepath{stroke}
        \pgfpathmoveto{\pgfpoint{0}{-0.5pt}}
        \pgfpathlineto{\pgfpoint{0}{2.5pt}}
        \pgfusepath{stroke}
        \end{pgfonlayer}
    }
}
\pgfdeclareshape{reverse arrow shape}
{
    \savedanchor\centerpoint
    {
        \pgf@x=0pt
        \pgf@y=0pt
    }
    \anchor{center}{\centerpoint}
    \anchor{south}{\centerpoint}
    \anchor{north}{\centerpoint}
    \anchorborder{\centerpoint}
    \backgroundpath
    {
        \begin{pgfonlayer}{foreground}
        \pgfsetstrokecolor{black}
        \pgfsetarrowsstart{<[\strarr]}
        \pgfsetlinewidth{\thickness}
        \pgfsetlinewidth{\thickness}
        \pgfpathmoveto{\pgfpoint{0}{-2pt}}
        \pgfpathlineto{\pgfpoint{0}{-1pt}}
        \pgfusepath{stroke}
        \end{pgfonlayer}
    }
}
\pgfdeclareshape{reverse double arrow shape}
{
    \savedanchor\centerpoint
    {
        \pgf@x=0pt
        \pgf@y=0pt
    }
    \anchor{center}{\centerpoint}
    \anchor{south}{\centerpoint}
    \anchor{north}{\centerpoint}
    \anchorborder{\centerpoint}
    \backgroundpath
    {
        \begin{pgfonlayer}{foreground}
        \pgfsetstrokecolor{black}
        \pgfsetarrowsstart{<[\strarr]}
        \pgfsetlinewidth{\thickness}
        \pgfpathmoveto{\pgfpoint{0}{-3.5pt}}
        \pgfpathlineto{\pgfpoint{0}{-0.5pt}}
        \pgfusepath{stroke}
        \pgfpathmoveto{\pgfpoint{0}{-0.5pt}}
        \pgfpathlineto{\pgfpoint{0}{2.5pt}}
        \pgfusepath{stroke}
        \end{pgfonlayer}
    }
}
\makeatother

\tikzset{arrow/.pic={
    \path[pic actions, decoration={ markings,
      mark=at position 1 with {\arrow{>[\strarr]}}
    },
    postaction={decorate}] (0,1pt) -- (0,2pt);
},double arrow/.pic={
    \path[pic actions, decoration={ markings,
      mark=at position \pgfdecoratedpathlength-3pt with {\arrow{>[\strarr]}},
      mark=at position \pgfdecoratedpathlength with {\arrow{>[\strarr]}}
    },
    postaction={decorate}] (0,-6pt) -- (0,4pt);
},
  reverse arrow/.pic={
    \path[pic actions, decoration={ markings,
      mark=at position 1 with {\arrow{<[\strarr]}}
    },
    postaction={decorate}] (0,1pt) -- (0,2pt);
}
}

\tikzset{surface picture/.style={xscale={0.6}, yscale=0.8, line width=\thickness}}
\tikzset{three dimensional picture/.style={}}
\tikzset{morphismtwocell/.style={morphism, width=0.5cm, connect ne length=0cm, connect se length=0cm, connect nw length=0cm, connect sw length=0cm}}
\tikzset{nmorphismtwocell/.style={draw, minimum width=0.8cm, connect ne length=0cm, connect se length=0cm, connect nw length=0cm, connect sw length=0cm}}

\tikzset{dashback/.style={black!40}}

\tikzset{shade 1 local/.style={shade 1}}
\tikzset{shade 2 local/.style={shade 2}}

%% file: subtexs/01_Introduction.tex
\section{Introduction}

    For much of cryptography's history, security has been assumed but not proven. Even today, we rely on protocols without proven security, which are rather based on intuitive arguments \cite{Katz_Lindell}. In the comparably young field of quantum cryptography, many protocols claim provable security under the assumption that the devices used in these protocols are trustworthy. While the protocols offer a real advantage in tackling modern cryptographic challenges \cite{Quantum_crypto_rev,Pirandola:20}, they often come with two caveats:
    \begin{enumerate}
        \item Quantum hardware is expensive and difficult to operate and maintain. This is particularly true for quantum computers or their main building blocks, such as implementations of entangling gates \cite{Preskill_2018}.
        \item The devices might not be trustworthy. To assume otherwise might in fact be a very strong assumption; someone untrusted could be operating the device, or there could be a hardware-based attack that leaks important information, as was done in the past for quantum key distribution systems \cite{lydersen:hacking}.
    \end{enumerate}

   Interestingly, these two issues are connected. Indeed, one way to address the first issue is to delegate some complex tasks to other parties while making sure that these perform the tasks as requested. Quantum correlations provide a way to check that the operations and tasks at hand are executed correctly. In the most general case, this is done through a framework called 'Device-Independence' \cite{Ac_n_2007}, where the parties involved in a protocol can verify that the operations performed are correct, without putting any trust on the hardware. 
    
    In this paper, we focus on one specific task: quantum state verification. In quantum state verification protocols, an untrusted source prepares quantum states and distributes them among the clients who are sometimes considered honest. If the source is honest, it always prepares the target state, i.e. the state the clients desire to hold, and the clients accept the result. However, if the source is dishonest, it might not always send the target state, and the clients should ideally reject it. By virtue of the no-cloning theorem, the clients cannot simply measure and then use the quantum states the source sent. Hence the need to verify that the quantum states are indeed correct. A typical way to verify quantum states is for the source to send several copies of the state, some of which are then measured by the clients. If enough measurements correspond to the expected the states,the clients are convinced that the source is honest. They can then use the states they did not measure for further tasks. Indeed, some protocols use quantum state verification as a subroutine, for example when the clients don't have the local resources to create the state or the network resources to distribute it \cite{dejong:anonymousCKA}. This modular use of quantum state verification begs for composable security but until recently, it wasn't clear if a quantum state verification protocol could be composably secure, especially if the clients are not trusted. Then, in \cite{Yehia_2021}, the authors showed composable security for the protocol in \cite{Pappa_2012} but only against a dishonest source. Following a different approach for the clients, the authors of \cite{colisson2024graph} demonstrated that stand-alone security implies composable security for many target states.\\ 
    In this work we provide a no-go result showing that a quantum state verification protocol cannot be composably secure and efficient at the same time. We use the novel framework of categorical composable cryptography \cite{broadbentkarvonen:categoricalcrypto,broadbentkarvonen:categoricalcryptoextended} to prove this result. 
    The motivation for using this framework lies in its combination of rigor and flexibility. On the one hand, modeling quantum processes and protocols as morphisms in a category provides a precise, albeit abstract, machine model which, by design, prevents mistakes and hidden assumptions. On the other hand, the ability to define attack models in the framework formally but still freely allows for the flexibility to investigate more complex adversarial settings such as 'honest-but-curious' or notions of i.i.d.-restrictions. Although we use a rather general attack model for our result, we will see that by the nature of the actual attacks, one could use more restrictive attack models as well. However, we stress that we believe that our results could be proven with essentially the same proofs in other frameworks for composably secure quantum cryptography~\cite{Unruh10,MauRen11,Mau11}.

\subsection{Our contribution and related work}
Many protocols implement quantum state verification for different types of states, e.g. \cite{hypergraph,many_qubit,Pappa_2012,Unnikrishnan_2022}. However, all protocols suffer from the same efficiency vs. security trade-off: a quantum state verification protocol cannot be secure and efficient. We investigate this trade-off in a general setting and find fundamental limitations for quantum state verification.
\begin{theorem}[Main result (informal)]
    Let $\pi$ be a protocol for quantum state verification with the following properties:
    \begin{itemize}
        \item the clients cannot prepare the target state, and
        \item if the clients output a state received from the source, they perform no map on it.
    \end{itemize}
    At least one of the following statements about $\pi$ with security parameter $\lambda$ is false:
    \begin{enumerate}
        \item $\pi$ rejects the target state with a probability negligible in $\lambda$.
        \item If the source is dishonest, either the probability to accept or the distance to the target state is negligible in $\lambda$.
        \item The number of rounds $N$ is polynomial in $\lambda$.
    \end{enumerate}
    Moreover, we find with $\varepsilon_H$ being the distinguishability to the idealized process if the source is honest and $\varepsilon_D$ if it is not
    \begin{align*}
        \varepsilon_H+\varepsilon_D\geq \begin{cases}
            \nicefrac{1}{8\sqrt{N}}\text{ if the target state is pure.}\\
            \nicefrac{1}{27N}\text{ if the target state is mixed.}\\
        \end{cases}
    \end{align*}
\end{theorem}
This trade-off has been proven in other works before, albeit in more restrictive settings. In both \cite{Bound1} and \cite{Bound2}, the authors showed that for a quantum verification protocol for pure target states with a fixed number of rounds, the worst-case infidelity ($1-Fid.$) scales with the inverse of number of rounds. Although, in \cite{Bound1} the authors argue that this is not a restriction, both assume that the clients perform single round tests, i.e. do not use collective measurements. 
However, our work differs from \cite{Bound1,Bound2} in many aspects. First, the assumptions differ: we do not consider a fixed number of rounds, we allow for collective measurements, and, very importantly, we derive a bound for mixed states as well.
To illustrate the importance of the latter, consider quantum state verification for a pure target state $\ketbra{\phi}$. One can then circumvent the results from \cite{Bound1,Bound2} by using a protocol with a target state $(1-f(N))\ketbra{\phi} + f(N) \ketbra{\phi^{\bot}}$, where $f$ can be even negligible in the number of rounds $N$. Our result closes such loopholes.
Further, the perspectives on the topic are different. In \cite{Bound1,Bound2}, the authors utilize the hypothesis testing framework\footnote{See \cite{QSV-review} for a review on quantum state verification focused on the hypothesis testing approach.}, which is useful for quantum state verification but is not common in other areas of quantum cryptography. We argue that quantum state verification should be viewed as a building block of larger protocols and hence use categorical composable cryptography. Because of this difference, we developed a novel proof technique which we expect to also be adaptable to other settings.\\[1em]
Our results provide bounds for self-testing as well \cite{Self-testing-review}. Self-testing is slightly different from quantum state verification, since there is a single client that does not trust any of their devices, including preparation and measurement apparatus. Self-testing can therefore be seen as a stricter case of quantum state verification. Hence, any attack on quantum state verification implies an attack on self-testing.

\subsection{Structure}
Our work is structured as follows. In section 2, we first present results from quantum information theory that we need for our security analysis. We then provide a gentle introduction to category theory and our category $\CPTP$ alongside morphisms used to express the algorithms in later sections. We conclude section 2 with the $\ncomb$ construction necessary to define the resource theories for the subsequent work. We outline the resource theory that we work with in section 3, guided by \cite{broadbentkarvonen:categoricalcrypto,broadbentkarvonen:categoricalcryptoextended}. We then give a formal security definition and a description of our ideal resource before presenting our main results in section 4. In section 4, we present our no-go result first for a simple type of protocols and then for general quantum state verification protocols. In both cases, we prove the no-go result in the single- and multi-client cases. Finally, we discuss open questions and possible implications of our work in section 5.

%% file: subtexs/02_Methods.tex
\section{Preliminaries}
\subsection{Quantum Information Theory} \label{sec:QITPrelim}
Before instantiating the categorical framework we need for our security analysis, we present some preliminaries on quantum information theory. All results in this section are taken from~\cite{watrous_2018}. 

In the following, we write our definitions with respect to density operators and quantum channels, although they hold for general linear operators and linear maps. A density operator on a space $\X$ is a positive semidefinite operator with trace equal to one. $D(\X)$ is the space of density operators. A quantum channel from $\X$ to $\Y$ is a completely positive trace preserving map from the space of linear operators on $\X$, $\L(\X)$ to $\L(\Y)$. $\C(\X,\Y)$ denotes the space of quantum channels from $\X$ to $\Y$.

\begin{definition}[Trace norm, diamond norm, diamond distance]
    For a density operator $\rho\in D(\X)$, we define the trace norm to be
$\|\rho\|_1=\Tr(\sqrt{\rho\rho\D}).$ The induced trace norm of a quantum channel $\Phi\in \C(\X,\Y)$ is then
$\|\Phi\|_1=\max\{|\Phi(\rho)|_1:\rho\in D(\X)\}$. The diamond norm of $\Phi$ is then defined as
\begin{equation*}
    \|\Phi\|_\diamond=\|\Phi\otimes \id[L(\X)]\|_1.
\end{equation*}
Finally we define the diamond distance between $\Phi,\Psi\in \C(\X,\Y)$ as
\begin{equation}\label{def:diamond_distance}
    \dia(\Phi,\Psi)=\|\Phi-\Psi\|_\diamond.
\end{equation}
\end{definition}

Two properties of the diamond distance that hold for CPTP maps $\Phi_0,\Phi_1,\Psi_0,\Psi_1\in \C(\X,\Y)$ and any space $\Z$ are
\begin{align}
    &\dia\left(\Psi_1\Psi_0,\Phi_1\Phi_0\right)\leq \dia(\Psi_1,\Phi_1)+\dia(\Psi_0,\Phi_0)\label{eq:diamond_distance_composition},\\
    &\dia\left(\Phi\otimes\id[L(\Z)],\ \Psi\otimes\id[L(\Z)]\right)\leq \dia(\Phi,\Psi).\label{eq:diamond_dist_with_id}
\end{align}

The distances above also have an operational interpretation. Indeed the trace distance yields a bound on the achievable distinguishing advantage between two density operators given by the Holevo-Helstrom Theorem.

\begin{theorem}[Holevo-Helstrom Theorem]
Let $\rho_0, \rho_1\in D(\X)$ be density operators, and let $\lambda\in [0, 1]$. For any measurement $\mu:\{0, 1\}\rightarrow Pos(\X)$ (where $Pos(\X)$ denotes all positive definite operators on $\X$) it then holds
\begin{equation}\label{eq:holevo_helstrom}
    \lambda \braket{\mu(0)}{\rho_0}+ (1-\lambda) \braket{\mu(1)}{\rho_1}\leq \frac{1}{2}+\frac{1}{2}\|\lambda\rho_0-(1-\lambda)\rho_1\|_1.
\end{equation}
Moreover there exists a projective measurement $\mu:\{0, 1\}\rightarrow Pos(\X)$ for which equality is achieved in (\ref{eq:holevo_helstrom}).
\end{theorem}

To see that this actually gives a bound on the distinguishing advantage we set $\lambda=\frac{1}{2}$ in (\ref{eq:holevo_helstrom}) and we obtain
\begin{align}
        &\frac{1}{2} \braket{\mu(0)}{\rho_0}+  \frac{1}{2} \braket{\mu(1)}{\rho_1}\leq \frac{1}{2}+ \frac{1}{4}\|\rho_0-\rho_1\|_1\\
        \Leftrightarrow& \braket{\mu(1)}{\rho_1} + (\braket{\mu(0)}{\rho_0}-1) = \braket{\mu(1)}{\rho_1} -\braket{\mu(1)}{\rho_0}\leq \frac{1}{2}\|\rho_0-\rho_1\|_1.\label{eq:dist_adv}
\end{align}
Using (\ref{eq:dist_adv}) with an adequate choice of measurement $\mu$ that satisfies $\braket{\mu(0)}{\rho_0}\geq\frac{1}{2}$ and $\braket{\mu(1)}{\rho_1}\geq\frac{1}{2}$, we find the distinguishing advantage. 

Another important quantity is the fidelity. The fidelity between two density operators $\rho_0$, $\rho_1$ is given by
\begin{equation*}
F(\rho_0, \rho_1)=\Tr\left(\sqrt{\sqrt{\rho_0}\rho_1\sqrt{\rho_0}}\right)^2.
\end{equation*}

The Fuchs-van de Graaf inequalities link the trace distance to the fidelity~\cite{watrous_2018}. 
\begin{theorem}[Fuchs-van de Graaf Inequalities] Let $\rho_0, \rho_1\in D(\X)$ be density operators, it  holds that
\begin{align}
   &1-\sqrt{F(\rho_0, \rho_1)}\leq \frac{1}{2}\|\rho_0-\rho_1\|_1\leq
\sqrt{1-F(\rho_0, \rho_1)}\label{eq:FvdG_1}\\
&\left(1-\frac{1}{2}\|\rho_0-\rho_1\|_1\right)^2\leq F(\rho_0, \rho_1)\leq 1-\frac{1}{2}\|\rho_0-\rho_1\|_1^2\label{eq:FvdG_2}
\end{align}
\end{theorem}

\begin{lemma}[Bounds for multi copy distinction] \label{lemma:Multi_copy_dist}For $\rho_0, \rho_1\in D(\X)$ we have
\begin{equation}\label{eq:multi_copy_dist}
        1-\left(\sqrt{1-\frac{1}{2}\|\rho_0-\rho_1\|_1^2}\right)^n\leq\frac{1}{2}\|\rho_0^{\otimes n}-\rho_1^{\otimes n}\|_1\leq \sqrt{1 - (1-\|\rho_0-\rho_1\|_1)^{n}}.
    \end{equation}
\end{lemma}
\begin{proof}
An important property of the fidelity function is
    \begin{equation}\label{eq:fidelity_n}
        F(\rho_0^{\otimes n}, \rho_1^{\otimes n})=F(\rho_0, \rho_1)^n.
    \end{equation}
    Using this and the Fuchs-van de Graaf inequalities we first derive the left hand side of (\ref{eq:multi_copy_dist})
    \begin{align*}
        &1-\left(\sqrt{1-\frac{1}{2}\|\rho_0-\rho_1\|_1^2}\right)^n\stackrel{(\ref{eq:FvdG_1})}{\leq} 1-\sqrt{F(\rho_0, \rho_1)}^n\\
        \stackrel{(\ref{eq:fidelity_n})}{=}&1-\sqrt{F(\rho_0^{\otimes n}, \rho_1^{\otimes n})}
        \stackrel{(\ref{eq:FvdG_1})}{\leq}\frac{1}{2}\|\rho_0^{\otimes n}-\rho_1^{\otimes n}\|_1.
    \end{align*}
   For the right hand side we have
    \begin{align*}
        &\frac{1}{2}\|\rho_0^{\otimes n}-\rho_1^{\otimes n}\|_1\stackrel{(\ref{eq:FvdG_1})}{\leq}\sqrt{1-F(\rho_0^{\otimes n}, \rho_1^{\otimes n})}\\
        \stackrel{(\ref{eq:fidelity_n})}{=}&\sqrt{1-F(\rho_0, \rho_1)^n}\stackrel{(\ref{eq:FvdG_2})}{\leq}\sqrt{1 - (1-\frac{1}{2}\|\rho_0-\rho_1\|_1)^{2n}}\\
        &\leq\sqrt{1 - (1-\|\rho_0-\rho_1\|_1)^{n}}.
    \end{align*}\qed
\end{proof}
 For pure states however, it holds that \begin{equation*}
     \|\ketbra{\psi}-\ketbra{\phi}\|_1 = 2\sqrt{1-|\braket{\psi}{\phi}|^2},
 \end{equation*}which implies
    \begin{align}\label{eq:pureState}
        \frac{1}{2}\|\rho_0^{\otimes n}-\rho_1^{\otimes n}\|_1 = \sqrt{1-|\braket{\psi}{\phi}|^{2n}}.
    \end{align}

\subsection{Category Theory}


To introduce the framework of categorical composable cryptography, we need to introduce some notions about category theory. Informally, a category is a collection of objects - usually denoted $A, B, C, \ldots$ - and morphisms - $f, g, h, \ldots$ - between objects.
Whenever we have two morphisms $f : A \to B$, $g : B \to C$ such that the domain of $g$ and the codomain of $f$ coincide, we can compose them to obtain the morphism $g \circ f : A \to C$. This composition operation is required to be associative and for every object $A$, there should exist a morphism $\id[A]$ which acts as identity on morphism composition. 

\begin{example}\label{ex:cats}
    \begin{enumerate}
        \item  Sets and functions between them form a category, $\cat{Set}$, in which the objects are sets and the morphisms are functions. Morphism composition is function composition and the identity morphism is the identity function, $f(x) = x$.
        \item    The category $\cat{FHilb}$, is the category in which the objects are finite dimensional Hilbert spaces and the morphisms are linear transformations between them.
        \item     The category $\cat{Met}$ of extended pseudometric spaces has extended pseudometric spaces as its objects: these are pairs $(X,d)$ where $X$ is a set and $d\colon X\times X\to [0,\infty]$ satisfies the axioms of a pseudometric\footnote{These are almost the axioms of a metric, except distinct points can have distance zero.}, except that we allow for points with infinite distance\footnote{This corresponds to the adjective ``extended'', and is mostly for mathematical convenience. This can be ignored as in the sequel as all metrics we use take finite values.}. The morphisms in $\cat{Met}$ are given by the \emph{short} (or distance non-increasing) maps, so that maps $(X,d)\to (Y,e)$ are given by functions $f\colon X\to Y$ satisfying $e(f(x),f(y))\leq d(x,y)$ for all $x,y\in X$. 
        \item  Recall that a monoid is basically a group without inverses, \ie a set $M$ equipped with a binary operation $\cdot\colon M\times M\to M$ that is associative and has a unit element. Any monoid $(M,\cdot)$ can be viewed as a category with one object $\bullet$, with the morphisms $\bullet\to\bullet$ given by elements of $M$ and composition given by $\cdot$.
        \item   Any partially ordered set $(P,\leq)$ induces a category whose objects are given by the elements of $P$, and there exists a
        a unique morphism $x\to y$ iff $x\leq y$.
    \end{enumerate}
\end{example}

    For any two objects in a category $\CC$, we denote the set of all morphisms $A\to B$ by $\CC (A,B)$.

In many categories of interest one can not only compose morphisms sequentially, but also in parallel. For instance in $\cat{Set}$, given two morphisms $f\colon A\to B$ and $g\colon C\to D$, we can form the morphism $f\times g\colon A\times C\to B\times D$. This parallel composition is almost associative, commutative and has a unit. For instance, there is an obvious bijection relating the sets $A\times (B\times C)$ and $(A\times B)\times C$ that merely re-brackets the data. This idea is made precise by the notion of a \emph{symmetric monoidal category} (SMC). We begin by introducing a stricter notion which is easier to define precisely although it fails to capture many examples of interest.

\begin{definition}[Symmetric strict monoidal category]

 A strict monoidal category $(\CC, \otimes, I)$ is a category $\CC$ equipped with an object $I$ called the monoidal unit and a monoidal product $\otimes$ sending a pair $(A,B)$ of objects to an object $A\otimes B$, and two morphisms $f\colon A\to B$ and $g\colon C\to D$ 
 to a morphism $f\otimes g\colon A\otimes C\to B\otimes D$. The operation $\otimes$ must respect identity morphisms in that $\id[A]\otimes \id[B]=\id[A\otimes B]$. Moreover, the operation $\otimes$ should satisfy the \emph{interchange law}, which states 
 that whenever $g\circ f$ and $i\circ h$ are defined, then 
\begin{equation}\label{eq:interchange}
(g \circ f) \otimes (i \circ h) = (g \otimes i) \circ (f \otimes h)\end{equation}
 Finally, the operation $\otimes$ should be associative and unital in that for all objects $A,B,C$ and morphisms $f,g,h,$ we have:
\begin{equation*}
\begin{aligned}
    &(A \otimes B) \otimes C = A \otimes (B \otimes C),\\
    &I \otimes A = A \otimes I = A ,\\
    &(f \otimes g) \otimes h = f \otimes (g \otimes h),\\ 
    &\id[I] \otimes f = f \otimes \id[I] = f.
\end{aligned}
\end{equation*}
A symmetric strict monoidal category is a strict monoidal category with chosen isomorphisms $\sigma_{A,B} : A \otimes B \to B \otimes A$ for all $A,B$ such that (i) $\sigma_{B,A} \circ \sigma_{A,B} = \id[A \otimes B]$ (ii) the isomorphisms $\{\sigma_{A,B}\}_{A,B}$ are natural in the sense that $(g\otimes f)\circ \sigma = \sigma \circ (f\otimes g)$. 
\end{definition}

In the general (not necessarily strict) case, a monoidal category has an operation $\otimes$ as above, but the associativity and unit equations of it are replaced by isomorphisms (satisfying some further conditions), see \cite{heunenVicary_ct} for more details.

\begin{example}\label{ex:SMCs}
\begin{enumerate}
    \item    The monoidal structure of $\cat{Set}$ is described as follows: the monoidal product is the cartesian product, $\times$, and the unit object is a chosen one element set: $\{\bullet\}$.
    \item     The monoidal structure of $\cat{FHilb}$ has the tensor product of Hilbert spaces as the monoidal product and the one-dimensional Hilbert space, $\Cc$, as the unit object.
    \item     We equip $\cat{Met}$ with a monoidal structure as follows: we define $(X,d)\otimes (Y,e)$ to be the set $X\times Y$ equipped with the $\ell^1$ distance, so that the distance between $(x,y)$ and $(x',y')$ in $X\otimes Y$ is given by the sum $d(x,x')+e(y,y')$.
\end{enumerate}
\end{example}

A nice feature of (symmetric) monoidal categories is that there is an intuitive yet precise graphical syntax for describing morphisms in them. We next introduce these \emph{string diagrams} and our conventions for them. We draw string diagrams from left to right, just like quantum circuits, although the reader should be warned that different papers might have their string diagrams drawn from top to bottom or from bottom to top instead.
The sequential composition looks like:
\[\begin{pic}
\node[box=0/1/0/1] (C) at (5.15,0) {g \circ f};
\draw (C.east) to node[above] {$C$} ++(1, 0);
\draw (C.west) to node[above] {$A$} ++(-1, 0);
\end{pic}
\quad := \quad 
\begin{pic}
\node[box=0/1/0/1] (A) at (0,0) {f};
\node[box=0/1/0/1] (B) at (2,0) {g};
\draw (A.east) to node[above] {$B$} (B.west);
\draw (A.west) to node[above] {$A$} ++(-1, 0);
\draw (B.east) to node[above] {$C$} ++(1, 0);
\end{pic}\]

While the tensor composition is simply drawing two morphisms in parallel.
\[\begin{pic}
\node[box=0/2/0/2] (A) at (0,-0.75) {f \otimes g};
\draw (A.east.1) to node[above] {$B$} ++(1, 0);
\draw (A.west.1) to node[above] {$A$} ++(-1, 0);
\draw (A.east.2) to node[below] {$D$} ++(1, 0);
\draw (A.west.2) to node[below] {$C$} ++(-1, 0);
\end{pic}
\quad := \quad 
\begin{pic}
\node[box=0/1/0/1] (C) at (4.5,-0.25) {f};
\node[box=0/1/0/1] (D) at (4.5,-1.25) {g};
\draw (C.east) to node[above] {$B$} ++(1, 0);
\draw (C.west) to node[above] {$A$} ++(-1, 0);
\draw (D.west) to node[above] {$C$} ++(-1, 0);
\draw (D.east) to node[above] {$D$} ++(1, 0);
\end{pic}\]
The symmetry isomorphisms are drawn as wire crossings :
\[\begin{pic}
\node[box=0/2/0/2, minimum height=10mm] (A) at (0,0) {\sigma_{A,B}};
\draw (A.west.1) to node[above] {$A$} ++(-1, 0);
\draw (A.west.2) to node[below] {$B$} ++(-1, 0);
\draw (A.east.1) to node[above] {$B$} ++(1, 0);
\draw (A.east.2) to node[below] {$A$} ++(1, 0);
\end{pic}
\quad := \quad 
\begin{pic}
\node (botl) at (0, -0.5) {};
\node (topl) at (0, 0.5) {};
\draw (topl) to node[above] {$A$} ++(.75,0) to[in=180,out=0] ++(1.5,-1) to node[below] {$A$} ++(.75,0);
\draw (botl) to node[below] {$B$} ++(.75,0) to[in=180,out=0]  ++(1.5,1) to node[above] {$B$} ++(.75,0);
\end{pic}
\]

The string diagrams make the axioms intuitive: for instance, the condition $\sigma_{B,A} \circ \sigma_{A,B} = \id[A \otimes B]$ becomes 
\[\begin{pic}\node (botl) at (0, -0.5) {};
\node (topl) at (0, 0.5) {};
\draw (topl) to node[above] {$A$} ++(.75,0) to[in=180,out=0] ++(1.5,-1) to node[below] {$A$} ++(.75,0) to[in=180,out=0]  ++(1.5,1) to node[above] {$A$} ++(.75,0);
\draw (botl) to node[below] {$B$} ++(.75,0) to[in=180,out=0]  ++(1.5,1) to node[above] {$B$} ++(.75,0)  to[in=180,out=0] ++(1.5,-1) to node[below] {$B$} ++(.75,0);
\end{pic}
\quad=\quad \begin{pic}\node (botl) at (0, -0.5) {};
\node (topl) at (0, 0.5) {};
\draw (topl) to node[above] {$A$} ++(.75,0) to ++(1,0);
\draw (botl) to node[below] {$B$} ++(.75,0) to ++(1,0);
\end{pic}
\]
so that two crossings undo each other, and the naturality condition for the symmetry can be pictured as us being allowed to slide the boxes corresponding to the morphisms $f$ and $g$ through the crossing:
\[\begin{pic}
\node[box=0/1/0/1] (A) at (0, 0.5) {g};
\node[box=0/1/0/1] (B) at (0, -0.5) {f};
\draw (A.east.1) to node[above] {$D$} ++(1,0);
\draw (B.east.1) to node[above] {$B$} ++(1,0);
\draw (A.west.1) to[in=0,out=180] ++(-1.5,-1) to node[above] {$C$} ++(-.75,0);
\draw (B.west.1) to[in=0,out=180]  ++(-1.5,1) to node[above] {$A$} ++(-.75,0);
\end{pic}
\quad
=\quad
\begin{pic}
\node[box=0/1/0/1] (A) at (0, 0.5) {f};
\node[box=0/1/0/1] (B) at (0, -0.5) {g};
\draw (A.west.1) to node[above] {$A$} ++(-1,0);
\draw (B.west.1) to node[above] {$C$} ++(-1,0);
\draw (A.east.1) to[in=180,out=0] ++(1.5,-1) to node[above] {$B$} ++(.75,0);
\draw (B.east.1) to[in=180,out=0]  ++(1.5,1) to node[above] {$D$} ++(.75,0);
\end{pic}\]

Other axioms are used in the pictures implicitly: for instance, when drawing three parallel lines we don't add in any brackets, and the interchange law \eqref{eq:interchange} guarantees that the following picture is unambiguous:
\[\begin{pic}
\node[box=0/1/0/1] (A) at (0,1) {f};
\node[box=0/1/0/1] (B) at (0,0) {h};
\node[box=0/1/0/1] (C) at (1.5,1) {g};
\node[box=0/1/0/1] (D) at (1.5,0) {i};
\draw (C.east) to node[above] {} ++(1, 0);
\draw (A.west) to node[above] {} ++(-1, 0);
\draw (A.east) to (C.west);
\draw (B.west) to node[above] {} ++(-1, 0);
\draw (B.east) to (D.west);
\draw (D.east) to node[above] {} ++(1, 0);
\end{pic}\]

We will also need the definition of a bimonoidal (rig) category, which we state formally below. This is a ``categorified'' version of a rig/semiring ($\approx$ ring without negatives, like the natural numbers with addition and multiplication), just like a monoidal category is a categorified version of a monoid.

\begin{definition}
    A bimonoidal category is a category equipped with two monoidal products - a symmetric monoidal structure $(\CC, \oplus, 0)$ and a monoidal structure $(\CC, \otimes, I)$- such that there exist distributivity  isomorphisms 
    \begin{equation*}
    \begin{aligned}
    &d_l : A \otimes (B \oplus C) \to (A \otimes B) \oplus (A \otimes C),\\
    &d_r : (A \oplus B) \otimes C \to (A \otimes C) \oplus (B \otimes C),  
    \end{aligned}
    \end{equation*}
    and absorption isomorphisms
    \begin{equation*}
    \begin{aligned}
        &a_l : A \otimes 0 \to 0,\\
        &a_r : 0 \otimes A \to 0
    \end{aligned}
    \end{equation*}
    that satisfy some coherence laws \cite{laplaza_1972}. 
\end{definition}

\subsubsection{Base Categorical modeling}

We will model our quantum systems of interest as finite-dimensional $C^*$-algebras and our quantum processes as quantum channels \ie CPTP maps between them. We will only sketch these informally and refer the reader to~\cite{Keyl2002} for full details. A paradigmatic example of a $C^*$-algebra is given by the space of bounded operators on a Hilbert space. First of all this is a vector space over the complex numbers and composition of operators makes it into an algebra. Moreover, the operation of taking the adjoint equips this algebra with an involution, and the general notion of a $C^*$-algebra abstracts away from this by axiomatizing important interactions between these structures and the operator norm. It is standard that any $C^*$-algebra embeds into one of this form, much like any group embeds into a permutation group. 

In the finite-dimensional case, a paradigmatic example of a $C^*$-algebra is given by $M_n(\Cc)$, the $n\times n$ complex matrices. Any finite-dimensional $C^*$-algebra is isomorphic to a finite direct sum of such matrix algebras (see e.g.~\cite[Theorem III.I.1]{davidson1996c}), and hence can be captured by a list $[n_1,\dots n_k]$ of non-negative natural numbers specifying the dimension of each matrix algebra.

The main reason we work with general (but finite-dimensional) $C^*$ algebras is that they allow us to treat quantum and classical systems. For example, the state-space of a qubit is modeled by the $C^*$-algebra $M_2(\Cc)$, whereas the state-space of a classical bit is modeled by $M_1(\Cc)\oplus M_1(\Cc)\cong \Cc\oplus \Cc$. The act of destructively measuring a qubit in the standard basis is then represented by the CPTP-map $M_2(\Cc)\to \Cc\oplus \Cc$ acting by 
$\begin{pmatrix}
a & b \\
c & d
\end{pmatrix}\mapsto (a\ d)$, 
and a non-destructive measurement of a qubit could be modeled as a map $M_2(\Cc)\to M_2(\Cc)\otimes (\Cc\oplus \Cc)$.

\begin{definition}[\CPTP]\label{def:cat_cptp}
The category $\CPTP$ of quantum channels is defined as follows: the objects are finite-dimensional $C^*$-algebras and the maps are completely positive trace preserving maps.
\end{definition}

The category of $C^*$-algebras admits two monoidal structures, $\oplus$ and $\otimes$ given by the natural direct sum and direct product of the underlying vector space. Thus, if $A \cong M_{n_1}(\Cc) \oplus \ldots \oplus M_{n_k}(\Cc)$ and $B \cong M_{m_1}(\Cc) \oplus \ldots \oplus M_{m_p}(\Cc)$, then 
\begin{align*}
	A \oplus B &= M_{n_1}(\Cc) \oplus \ldots \oplus M_{n_k}(\Cc) \oplus M_{m_1}(\Cc) \oplus \ldots \oplus M_{m_p}(\Cc) \\
	A \otimes B &= M_{n_1 m_1}(\Cc) \oplus M_{n_1 m_2}(\Cc) \oplus \ldots M_{n_1 m_p}(\Cc) \oplus M_{n_2 m_1}(\Cc) \oplus \ldots M_{n_k m_p}(\Cc)
\end{align*}
In shorthand, this can be represented as 
\begin{align*}
	[n_1,\dots, n_k] \oplus [m_1, \ldots, m_p] &= [n_1,\dots n_k, m_1, \ldots, m_p] \\
	[n_1,\dots, n_k] \otimes [m_1, \ldots, m_p] &= [n_1 m_1, n_1 m_2, \ldots, n_k m_p]
\end{align*} 

With respect to the $\oplus$ product, the 0-dimensional $C^*$-algebra is the unit object, while with respect to the  $\otimes$ product, $\Cc \cong M_1 (\Cc)$ is the unit object.

\begin{lemma}[\CPTP~is bimonoidal]
	The category $\CPTP$ is bimonoidal with product operations $\oplus$ and $\otimes$.
\end{lemma} 
\begin{proof}
	See \cite[Definition 2.10]{huotstaton:channelscategorically}.
\end{proof}


\subsubsection{Pseudo-Code}

To express algorithms in this category in a simple fashion, we introduce translations from pseudo-code to morphisms in \CPTP. This translation relies on the bimonoidal structure of \CPTP. We start with Branch-up and Branch-down -- isomorphisms that essentially state that there are two ways of expressing classical distributions of quantum states.  
\begin{definition}[Branch-up/ Branch-down]
Let $A$ be an object in $\CPTP$. Then, we define the isomorphism $B^{up}_{n,A}$ using the unitors and distributors as follows:
\begin{equation*}
B^{up}_{n,A}: I^{\oplus n} \otimes A \xrightarrow[]{\simeq} (I^{\oplus n-1} \otimes A) \oplus (I \otimes A) \xrightarrow{\simeq} (I^{\oplus n-1} \otimes A) \oplus A \xrightarrow{\simeq} \ldots \xrightarrow{\simeq} A ^{\oplus n}.
\end{equation*}

So, $B^{up}_{n,A}: I^{\oplus n} \otimes A \xrightarrow[]{\simeq} A^{\oplus n}$. 

Symmetrically, we define $B^{down}_{n,A}: A^{\oplus n} \xrightarrow{\simeq} I^{\oplus n} \otimes A$.
    
\end{definition}

In particular, for algorithms in the context of verification, explicit branching is essential. We allow for branching using the following definition.    
\begin{definition}[If-Else]\label{def:elif}
    Let $A, B, C$ be objects in $\CPTP$ and let $f \colon A\to B$, $g\colon A\to C$ be morphisms in $\CPTP$. Then, $\mathsf{ife}(f,g)\colon  (I\oplus I) \otimes A \to B \oplus C$ is a morphism (channel) defined as:
    \begin{equation*}
    \mathsf{ife}(f,g) = (f \oplus g)\circ B^{up}_{2,A}
    \end{equation*}

    More generally, an if-else channel applied to $n$ arguments $f_1, \ldots, f_n$ (and corresponding to an if-then-else structure with $(n-2)$ else-if structures) is a morphism $\mathsf{elif} ((f_i)_{i=1}^n) \colon (I^{\oplus n}\otimes A\to \bigoplus_{i=1}^n B_i)$ defined as
    \begin{equation*}
            \mathsf{elif} ((f_i)_{i=1}^n) = \left(\bigoplus_{i=1}^n f_i\right)\circ B^{up}_{n,A},
    \end{equation*} where $f_i : A \to B_i$. 
\end{definition}

Using the symmetry of the category allows for swapping registers. To denote this formally in an algorithm, we introduce the corresponding pseudo-code.
\begin{definition}[Swap]
Given $A_1, \ldots, A_n$, objects in $\CPTP$, to define a map $\mathsf{swap_{k, l}} : A_1 \otimes \ldots \otimes A_k \otimes \ldots \otimes A_l \otimes \ldots \otimes A_n \to A_1 \otimes \ldots \otimes A_l \otimes \ldots \otimes A_k \otimes \ldots \otimes A_n$ we first let $\Sigma_{k, k+1} = \bigotimes_{i=1}^{k-1} I \otimes \sigma_{k, k+1} \otimes \bigotimes_{i = k+2}^{n} I$.
Then, 
\begin{equation*}
\mathsf{swap_{k, l}} = \Sigma_{k, l-1} \circ \ldots \circ \Sigma_{k, k+1} \circ \Sigma_{k, l} \circ \ldots \circ \Sigma_{l-2, l} \circ \Sigma_{l-1, l}. 
\end{equation*}
Graphically, for the case $n=3$, if we want to swap $A_1$ and $A_3$ the equation looks like (note that we do not draw the tensor units):
\begin{equation*}\begin{pic}
\node[box=0/2/0/2,minimum height=10mm,opacity=0] (A) at (0,0) {\sigma_{2, 3}};
\node[box=0/2/0/2,minimum height=10mm,opacity=0] (B) at (2,0.5) {\sigma_{1, 3}};
\node[box=0/2/0/2,minimum height=10mm,opacity=0] (C) at (4,0) {\sigma_{1, 2}};
\draw (A.east.2) to[in=0,out=180] (A.west.1) to node[above] {$A_2$} ++(-1, 0);
\draw (A.east.1) to[in=0,out=180] (A.west.2) to node[below] {$A_3$} ++(-1, 0);
\draw (B.east.2) to[in=0,out=180] (B.west.1) to node[above] {$A_1$} ++(-3, 0);
\draw (A.west.1) to node[above] {$A_2$} ++(-1, 0);
\draw (A.east.1) to node[above] {$A_3$} (B.west.2) to[in=180,out=0] (B.east.1);
\draw (A.east.2) to node[below] {$A_2$} (C.west.2) to[in=180,out=0] (C.east.1);;
\draw (B.east.1) to node[above] {$A_3$} ++(3, 0);
\draw (B.east.2) to node[above] {$A_1$} (C.west.1) to[in=180,out=0] (C.east.2);;
\draw (C.east.1) to node[above] {$A_2$} ++(1, 0);
\draw (C.east.2) to node[below] {$A_1$} ++(1, 0);
\end{pic}\end{equation*}
\end{definition}

\begin{definition}[Move-back]\label{def:move-back}
    Given $A_1, \ldots, A_n$ objects in $\CPTP$, the map $\mathsf{move-back}_{k, n} : A_1 \otimes \ldots \otimes A_k \otimes \ldots \otimes A_n \to A_1 \otimes \ldots \otimes A_{k-1} \otimes A_{k+1} \otimes \ldots \otimes A_n \otimes A_k$ is defined as:
    \begin{equation*}
    \mathsf{move-back}_{k, n} = \circ_{i = k}^{n} \mathsf{swap_{i, i + 1}}. 
    \end{equation*}
\end{definition}

At last, we need to implement a different kind of branching. So far, we are only representing explicit branching. Hence, one can learn from the outside which branch the program chose. However, if the meta-data of the state, especially the dimensionality, does not give away which branch was chosen, one can choose not to output this information. In this case, we shall delete this information after branching with $\mathsf{elif}$.

\begin{definition}[Forget-branch]\label{def:forget-branch} Given an object $A$ in  $\CPTP$, the morphism\\ $\mathsf{forget-branch}_{n, A} : A^{\oplus n} \to A$ is defined as follows:
\begin{equation*}
\mathsf{forget-branch}_{n, A} = (\text{Tr}_{I^{\oplus n}}\otimes \text{id}_{A})\circ B^{down}_{n, A} ,
\end{equation*} where $\text{Tr}_{(-)}$ is the map to the monoidal unit.
\end{definition}

\subsubsection{n-combs}

In the following, we wish to formalize settings where we can model useful cryptographic resources based on quantum channels shared between $n$ parties. In cryptographic protocols, each party acts locally on their system, and the parties interact with one another over multiple rounds. Our category therefore needs to model both local actions on systems as well as multiple rounds of interactions between parties. To this end, we present a slightly modified version of the $\ncomb$ category defined in~\cite[Definition 3.2]{broadbentkarvonen:categoricalcryptoextended}. In $\ncomb$ morphisms represent a given agent's part of a protocol. 
In contrast to the one in~\cite{broadbentkarvonen:categoricalcryptoextended}, our definition allows for settings, where the protocols don't necessarily use all shared resources.

 \begin{definition}\label{def:ncomb} Given an SMC $\CC$, the category $\ncomb(\CC)$ is defined as follows: objects of $\ncomb(\CC)$ are finite lists $(A_i,B_i)_{i=1}^m$ of pairs of objects of $\CC$. Morphisms are defined in two stages: For $p \leq m$, a morphism $(A_i,B_i)_{i=1}^m\to (C,D)$ is given by an injection $\imath\colon \{1,\dots , p\}\to\{1,\dots , m\}$ and a $p$-comb


\begin{equation*}\begin{pic}
\node[box=0/2/0/1,minimum height=16mm] (A) at (0,0) {g_0};
\node[box=0/1/0/1,dashed] (B) at (1.5,.4) {};
\node[box=0/2/0/2,minimum height=16mm] (C) at (3,0) {g_1};
\node (b) at (4.5,.4) {$\dots$};
\node (c) at (4.5,-.4) {$\dots$};
\node[box=0/2/0/2,minimum height=16mm] (D) at (6,0) {g_{p-1}};
\node[box=0/1/0/1,dashed] (E) at (7.5,.4) {};
\node[box=0/1/0/2,minimum height=16mm] (F) at (9,0) {g_p};
\draw (A.west) to node[above] {$C$} ++(-1,0);
\draw (A.east.1) to node[above] {$A_{\imath(1)}$} (B.west);
\draw (A.east.2) to node[above]{$Y_1$} (C.west.2);
\draw (B.east) to node[above]  {$B_{\imath(1)}$} (C.west.1);
\draw (C.east.1) to node[above] {$A_{\imath(2)}$} (b.west);
\draw (C.east.2) to node[above] {$Y_2$} (c.west);
\draw (D.west.2) to node[above] {$Y_{p - 1}$} (c.east);
\draw (D.east.1) to node[above] {$A_{\imath(p)}$} (E.west);
\draw (D.east.2) to node[above] {$Y_p$} (F.west.2);
\draw (E.east) to node[above] {$B_{\imath(p)}$} (F.west.1);
\draw (F.east) to node[above] {$D$} ++(1,0);
\draw (b.east) to node[above] {$B_{\imath(p-1)}$} (D.west.1);
\end{pic}\end{equation*}
 
 in $\CC$. Formally, a $p$-comb is an equivalence class of tuples $(g_0,\dots, g_p)$ of maps in $\CC$,  where $g_0\colon C\to A_{\imath(1)}\otimes Y_1$, $g_i\colon B_{\imath(i)}\otimes Y_i\to A_{\imath(i+1)}\otimes Y_{i+1}$ for $i=1\dots p-1$ and $g_p\colon B_{\imath(p)}\otimes Y_p\to D$ for some objects $Y_i$. Two such tuples are identified if, whenever one ``plugs the holes'' with maps of the form $Z_i\otimes A_{\imath(i)}\to Z_i\otimes B_{\imath(i)}$, the resulting maps in $\CC$ are equal.

 A morphism $(A_i,B_i)_{i=1}^m\to (C_j,D_j)_{j=1}^k$ is given by a function $f\colon \{1,\dots , m\}\to\{1,\dots , k\}$ and a morphism $(A_i,B_i)_{i\in f^{-1}(j)}\to (C_j,D_j)$ for each $j$. Composition is defined by nesting circuits into circuits, and the monoidal product is given by concatenation of lists.
 \end{definition}
 
 Note that the monoidal product in the underlying category $\CC$ is different from that of in $\ncomb(\CC)$.

%% file: subtexs/03_CCC.tex
\section{Categorical Composable Cryptography}

One of the main contributions of~\cite{broadbentkarvonen:categoricalcrypto,broadbentkarvonen:categoricalcryptoextended} are highly general composition theorems. These can be viewed as giving a blueprint for numerous models of composable cryptography: one gets a specific model by fixing each degree of freedom in the formalism. To fix these, one first needs to choose two SMCs $\DD$ and $\CC$, where $\DD$ models the protocols, and $\CC$  models the relevant kind of (computational) processes, which may or may not be more general than those given in $\DD$. One also needs to fix a map $\DD\to \CC$ of SMCs which interprets protocols into processes. One also needs to give a map out of $\CC$, which gives for each object (thought of as a system type) the resources of that type, and specifies how processes in $\CC$ act on these resources. If one requires perfect security, this operation $R$ can be modeled as a suitable kind of map of SMCs $\CC\to\cat{Set}$, so that in particular for each object $A$ we have a set $R(A)$ of resources of that type. If we want to model security up to (computational) indistinguishability, $R(A)$ should be equipped
with an equivalence relation, and if we want to do security up to some notion of distance, then $R(A)$ should be a (pseudo)metric space. The chain of maps $\DD\to\CC\to\cat{Set}$ (or $\DD\to\CC\to\cat{Met}$) then induces a \emph{resource theory} of correct conversions between resources. To add in a notion of security, one needs a further structure called an \emph{attack model} on $\CC$, which in a nutshell specifies the way that adversaries can force a protocol to deviate from its intended behavior. One can then form the SMC of (suitably correct) resource conversions that are secure against this attack model, and the fact that this results in an SMC is the heart of the composition theorem---secure conversions are closed under sequential and parallel composition.  

For a detailed exposition on how these resource theories arise we refer the reader to~\cite{broadbentkarvonen:categoricalcrypto,broadbentkarvonen:categoricalcryptoextended}. For a more general study of resource theories one can consult~\cite{coecke_mathematical_2016}.

In this section we adapt this framework for our analysis of quantum state verification protocols. We present the resource theories we work in for single- and multi-client verification protocols. Based on this we can present our security definition and finally we give a formal definition of our ideal resources.
\subsection{The relevant resource theories}

We wish to consider ``security up to $\varepsilon$'', so our mapping $R$ specifying the resources of a given type should be a (pseudo)metric space. In other words, we wish for $R$ to land in the SMC $\cat{Met}$ of extended (pseudo)metric spaces and short maps from examples~\ref{ex:cats} and \ref{ex:SMCs}.

In fact, we will define two different, albeit similar, resource theories for multi- and single-party quantum state verification. In quantum state verification, we consider a source that can perform any arbitrary quantum operation. This corresponds to the category $\CPTP$, which we denote from now on with $\CC$ to simplify notation.
In the single-party case we consider the receiving party, that wishes to verify the quantum state, to be able to measure the state, but not to create the state. We define $\DD$ to be the sub-SMCs of $\CC$ generated by morphisms that are destructive quantum measurements and by arbitrary maps between classical systems (which correspond to stochastic maps).

For multi-party verification, we again consider a source corresponding to $\CC$ along with $k$ clients each also acting in $\CC$. The clients can only act locally and cannot create entanglement with one another. This restriction is represented by the Cartesian product $\CC^k$. The resource theories of single- and multi-party quantum state verification respectively are induced by the maps
\begin{align} \ncomb(\DD\times\CC)&\rightarrow\ncomb(\CC)\rightarrow\Met\label{restheory_single}\\
\ncomb(\CC^k\times\CC)&\rightarrow\ncomb(\CC)\rightarrow\Met\label{restheory_multi}
\end{align}
The morphisms on the left are the monoidal functors induced by the ($k$-)fold tensor product $\DD\times\CC\hookrightarrow\CC$ and $\CC^k\times\CC\hookrightarrow\CC$. The second map is given by $\ncomb(\CC)(I,-)$, where $I$ is the tensor unit in $\ncomb(\CC)$. 

Let us now explain what these abstract definitions amount to concretely, starting from the simpler case of~\eqref{restheory_single}. We first unwind the definitions. An object of $\ncomb(\DD\times \CC)$ is given by a finite list $(A_i,B_i)_{i=1}^n$ of objects of $\DD\times \CC$, but we'll first focus on lists $(A,B)$ of length one. In turn, an object of $\DD\times \CC$ is a pair of objects: one of $\DD$ and one of $\CC$. Thus each $(A,B)$ is of the form $((A_1,A_2),(B_1,B_2))$, and one can then show that the map~\eqref{restheory_single} sends $((A_1,A_2),(B_1,B_2))$ to $\CPTP(A_1\otimes A_2,B_1\otimes B_2)$. It follows that a resource of type $((A_1,A_2),(B_1,B_2))$  is given by a bipartite quantum channel $A_1\otimes A_2\to B_1\otimes B_2$, where we think of the first input and output belonging to the first party (the verifier) and the second input and output belong to the second party (the source). More generally, a resource of type $(A_i,B_i)_{i=1}^n$ is a list of $n$ such bipartite channels. 

Given a starting resource $f\colon A_1\otimes A_2\to B_1\otimes B_2$ (of type $((A_1,A_2),(B_1,B_2))$ and a target resource $g\colon C_1\otimes C_2\to D_1\otimes D_2$ (of type $((C_1,C_2),(D_1,D_2))$, a resource conversion $f\to g$ can be depicted by two $1$-combs, one for each party, as in
\begin{equation}\label{pic:pairofcombs}\begin{pic}
\node[box=0/2/0/1,minimum height=16mm] (A) at (0,0) {g_0};
\node[box=0/1/0/1,dashed] (B) at (1.5,.4) {};
\node[box=0/2/0/2,minimum height=16mm] (C) at (3,0) {g_1};
\draw (A.west) to node[above] {$C_1$} ++(-1,0);
\draw (A.east.1) to node[above] {$A_1$} (B.west);
\draw (A.east.2) to node[above]{$Y_1$} (C.west.2);
\draw (B.east) to node[above]  {$B_1$} (C.west.1);
\draw (C.east) to node[above] {$D_1$} ++(1,0);
\end{pic} \text{ , }
\begin{pic}
\node[box=0/2/0/1,minimum height=16mm] (A) at (0,0) {h_0};
\node[box=0/1/0/1,dashed] (B) at (1.5,.4) {};
\node[box=0/2/0/2,minimum height=16mm] (C) at (3,0) {h_1};
\draw (A.west) to node[above] {$C_2$} ++(-1,0);
\draw (A.east.1) to node[above] {$A_2$} (B.west);
\draw (A.east.2) to node[above]{$Y_2$} (C.west.2);
\draw (B.east) to node[above]  {$B_2$} (C.west.1);
\draw (C.east) to node[above] {$D_2$} ++(1,0);
\end{pic}
\end{equation}
where we require that the first part, belonging to the verifier, lives in the category $\DD$ (i.e., that $g_0$ and $g_1$ are morphisms in $\DD$). This resource conversion is correct, exactly if, when applied to $f$, it produces $g$, \ie filling the hole in 
\begin{equation}\label{pic:bipartitecomb}\begin{pic}
\node[box=0/2/0/1,minimum height=16mm] (A) at (0,0) {g_0};
\node[box=0/2/0/2,dashed] (B) at (1.5, -1) {};
\node[box=0/2/0/2,minimum height=16mm] (C) at (3,0) {g_1};
\node[box=0/2/0/1,minimum height=16mm] (D) at (0,-2) {h_0};
\node[box=0/2/0/2,minimum height=16mm] (E) at (3,-2) {h_1};
\draw (A.west) to node[above] {$C_1$} ++(-1,0);
\draw (A.east.2) to node[above] {$A_1$} (B.west.1);
\draw (A.east.1) to node[above]{$Y_1$} (C.west.1);
\draw (B.east.1) to node[above]  {$B_1$} (C.west.2);
\draw (C.east) to node[above] {$D_1$} ++(1,0);
\draw (D.west) to node[above] {$C_2$} ++(-1,0);
\draw (D.east.1) to node[above] {$A_2$} (B.west.2);
\draw (D.east.2) to node[above]{$Y_2$} (E.west.2);
\draw (B.east.2) to node[above]  {$B_2$} (E.west.1);
\draw (E.east) to node[above] {$D_2$} ++(1,0);
\end{pic}
\end{equation}
with $f$ results in  $g$. Typically, but not necessarily, the resources used enable communication between the parties, so that one could then think of these pictures as depicting a 2-party 1-round protocol.
A more general resource conversion  $(f_1,\dots f_n)\to g$ is similar, except that (i) there's more holes in the pictures (corresponding to more rounds in the protocol) and (ii) the parties have to agree on the order they call the shared resources $f_i$ (which, in the case of communication, amounts to agreeing what kind of information is sent at each round). The parties can also agree to not use some of the shared resources. We note that since n-combs are finite and holes represent rounds, we are de facto setting an upper limit on the number of rounds. However, this does not pose a problem since the size of the n-comb can be arbitrarily chosen. In fact, any sensible model would not allow for an infinite number of rounds and would abort after a preset number of rounds.

We consider security of such protocols in the next subsection, and conclude this subsection by verifying carefully that  $\ncomb(\CC)(I,-)$ is indeed a map to $\Met$.

Objects in $\ncomb(\CC)$ are finite lists, and the tensor unit $I$ of $\ncomb(\CC)$ is the empty list. For $C,D\in \CC$ we have $(C,D)\in\ncomb(C)$ and
\begin{equation*}
    \ncomb(\CC)(I,(C,D))=\CC(C,D).
\end{equation*}
By endowing $\CC(C,D)$ with the diamond distance (\ref{def:diamond_distance}), we obtain a metric space $(\CC(C,D),\dia)$.
For an object $(A_i,B_i)_{i=1}^n$, one can show that $\ncomb(\CC)(I,-)$ maps it to a product space
\begin{align}\label{eq:product_hom_spaces}
    \ncomb(\CC)(I, (A_i,B_i)_{i=1}^n)=\CC(A_1,B_1)\times\dots\times\CC(A_n,B_n).
\end{align}
For this product space we use the monoidal structure of $\cat{Met}$, so it is equipped with the $\ell^1$-distance given by the sum of diamond distance on each of the hom spaces. Let $d:=\sum_{i=1}^n\dia$ denote said distance. The product space (\ref{eq:product_hom_spaces}) is then also a metric space $\left(\bigtimes_{i=1}^n \CC(A_i,B_i),d\right)$.
 
 Next, we show that n-combs induce morphisms in $\Met$, which are short maps.  A n-comb maps a list $(A_i,B_i)_{i=1}^l$ with $l\geq n$ to $(C,D)$ This induces a morphism
 \[
     \gamma:\ncomb(\CC)(I, (A_i,B_i)_{i=1}^l)\rightarrow \ncomb(\CC)(I, (C,D)).
 \]
that is, a morphism 
\[\CC(A_1,B_1)\times\dots\times\CC(A_n,B_n)\to \CC(C,D)\] 
that acts on $(a_1,\dots a_n)$ by filling the holes in the comb with $a_i$.
We then need to check that $\gamma$ itself induces a short map, \ie a morphism in $\Met$. Let $x=(\imath,(h_i)_{i=0}^n)$ with $h_0:C\to A_{\imath(1)}\otimes Y_1,\ h_1:B_{\imath(i)}\otimes Y_i\to A_{\imath(i+1)}\otimes Y_{i+1}$ for $i=1,\dots,n-1$ and $h_n:B_{\imath(n)}\otimes Y_n\to D$ specify an n-comb, and let $x(\bar a)$ be the n-comb filled with the tuple of CPTP maps $\bar a =(a_1,\dots,a_l)$, where $a_i:A_i\rightarrow B_i$ and $l\geq n$. And $Y_k$ denotes an auxiliary register. To show that $\gamma$ induces a short map we need to show that
$d(\bar a, \bar b)\geq\dia(x(\bar a),x(\bar b))$. Indeed using the notation as in Definition \ref{def:ncomb}, with $\imath$ being an injection from $\{1,...,n\}$ to $\{1,...,l\}$, we can write $x(\bar a)$ as
\begin{align*}
    x(\bar a)=h_n\circ_{i=1}^n[(a_{\imath(i)}\otimes\id[Y_i])\circ h_{i-1}].
\end{align*}
By the properties of the diamond distance (\ref{eq:diamond_distance_composition}),(\ref{eq:diamond_dist_with_id}) and the fact that $\dia(f,f)=0$ by virtue of it being a metric, it then follows
\begin{align*}
    &\dia(x(\bar a),x(\bar b))\\
    =&\dia (h_n\circ_{i=1}^n[(a_{\imath(i)}\otimes\id[Y_i])\circ h_{i-1}], h_n\circ_{i=1}^n[(b_{\imath(i)}\otimes\id[Y_i])\circ h_{i-1}])\\
    \leq&\dia(h_n,h_n)+\dia (\circ_{i=1}^n[(a_{\imath(i)}\otimes\id[Y_i])\circ h_{i-1}],\circ_{i=1}^n[(b_{\imath(i)}\otimes\id[Y_i])\circ h_{i-1}])\\
    \leq&\sum_{i=1}^n\dia ((a_{\imath(i)}\otimes\id[Y_i])\circ h_{i-1},(b_{\imath(i)}\otimes\id[Y_i])\circ h_{i-1})\\
    \leq&\sum_{i=1}^n[\dia((a_{\imath(i)}\otimes\id[Y_i]),(b_{\imath(i)}\otimes\id[Y_i]))+\dia(h_{i-1},h_{i-1})]\\
     \leq&\sum_{i=1}^n \dia(a_{\imath(i)},b_{\imath(i)})
      \leq \sum_{i=1}^l \dia(a_i,b_i)=d(\bar a,\bar b).
\end{align*}
\subsection{Security Definition} 
Using the resource theories (\ref{restheory_single}) and (\ref{restheory_multi}), we can present our security definition based on \cite{broadbentkarvonen:categoricalcrypto,broadbentkarvonen:categoricalcryptoextended}. The security definition relies on an attack model $\A$. An attack model on a category $\CC$ gives for every morphism $f$ in $\CC$ a class of morphisms $\A(f)$ that fulfills certain properties which are stated in the original work.\\ 
To capture both of the situations above in a single situation, we model the situation with $K+1$ parties, where the last party acts maliciously in $\CC$, and let $\cat{E}$ be a sub-SMC of $\CC$ (where in the above, we either have $\cat{E}=\CC$ or $\cat{E}=\DD$).
We define an attack model on $\ncomb(\cat{E}^K\times \CC)$  derived from a general attack model in \cite{broadbentkarvonen:categoricalcryptoextended}.
The attack model $\A$ consists of allowing the last party to change their part of any $m$-comb arbitrarily while leaving everything else in the morphisms of $\ncomb(\cat{E}^K\times\CC)$ unchanged. For example, in the $2$-party $1$-round case depicted in~\eqref{pic:pairofcombs} and~\eqref{pic:bipartitecomb}, this amounts to allowing the second party to change their comb (and hence their resulting input $C_2$ and output $D_2$) arbitrarily, provided they do send and receive something of type $A_2$ and $B_2$ respectively into the shared resource.

We now give a formal definition of the attack model, but note that the intuitive definition above is sufficient for many purposes.

\begin{definition}[Attack model $\A$ on $\ncomb(\cat{E}^K \times \CC)$ ]\label{def:sourceAttack}
 We define an attack model $\A$ on $\ncomb(\cat{E}^K \times \CC)$ corresponding to $K$ honest parties and one malicious party. Consider a morphism in $\ncomb(\cat{E}^K \times \CC)$ given by an injection $\imath:\{1,\dots,m\}\rightarrow\{1,\dots,l\}$ and a $m$-tuple of morphisms in $\ncomb(\CC^k)$, $(g_0,\dots,g_m)$, that is 
\begin{equation*}
    (\imath,(g_0,\dots,g_m)):(A_i,B_i)_{i=1}^l\rightarrow (C,D)
\end{equation*}
Each $g_j$ is a morphism in $\cat{E}^K\times \CC$, and as such itself a tuple of morphisms in $\cat{E}$ and one in $\CC$. We write $\pi_j:\V{E}^K \times \CC \rightarrow\cat{E}$ for the $j$-the projection. The attack model is then defined as
\begin{equation*}
    \A((\imath,(g_0,\dots, g_m))):=\{(\imath,(h_0,\dots,h_m))|\pi_j(h_\ell)=\pi_j(g_\ell)\textrm{ for all }\ell\textrm{ and }1\leq j\leq K\}.
\end{equation*}
\end{definition}
In the following we study quantum state verification protocols with $K\geq 1$ honest clients and a dishonest source. The source is then modeled in $\CC$ and the clients correspond to $\cat{E}^K$ in Definition~\ref{def:sourceAttack}. 

Intuitively, security against a dishonest source means that for every attack $a$ in the attack model in $\A(\pi)$ applied on the protocol, there is an attack $b$ in the attack model $\A(\id[s])$ on the identity for the ideal resource such that $a$ applied on the real resources and $b$ on the ideal resource are indistinguishable up to $\varepsilon$. With this intuition and the attack model we defined before, we now define security formally.
\begin{definition}[Security against the source]
Let $\V{E}$ be any sub-SMC in $\CC$. We further consider $F:\ncomb(\cat{E}^K\times \CC)\rightarrow \ncomb(\CC)$ being the injection $\V{E} \hookrightarrow \CC$ followed by the $(K+1)$-fold tensor product and $R:\ncomb(\CC)\rightarrow \V{Met}$ given by $\ncomb(\CC)(I,-)$. A protocol (morphism in $\ncomb(\cat{E}^K\times \CC)$) $\pi:(A,\bar r) \rightarrow (B,\bar q)$ $\varepsilon$-securely implements $(B,\bar s)$ with an untrusted source if
\begin{align*}
    \forall a\in \A(\pi)\ \exists b\in \A(\id[B]): \hdia{RF(a) \bar r}{RF(b) \bar s}\leq \varepsilon,
\end{align*}
where $\A$ is the attack model as defined in Def. \ref{def:sourceAttack}.
\end{definition}
Apart from security, we need a definition of correctness, i.e. that the implementation is close to the ideal resource if all parties act honestly. 
\begin{definition}[Correctness]
    Let $\cat{E}$ be any sub-SMC in $\CC$. We further consider $F:\ncomb(\cat{E}^K\times \CC)\rightarrow \ncomb(\CC)$ being the injection $\V{E} \hookrightarrow \CC$ followed by the $(K+1)$-fold tensor product and $R:\ncomb(\CC)\rightarrow \cat{Met}$ given by $\ncomb(\CC)(I,-)$. A protocol (morphism in $\ncomb(\cat{E}^K\times \CC)$) $\pi:(A,\bar r) \rightarrow (B,\bar q)$ $\varepsilon$-correctly implements $(B,\bar s)$ if
\begin{align*}
     \hdia{RF(\pi) \bar r}{RF(\id[B]) \bar s}\leq \varepsilon.
\end{align*}
\end{definition}

\subsection{Ideal resource}

To prove that there is no efficient and secure quantum state verification, we also need to define the ideal resource. The ideal resource we consider is the same as in \cite{colisson2024graph}.\\ 
To ease the reading of the ideal resource as morphisms, we will write out tensor units explicitly when they represent the input or output of parties. Further, we use the tensor unit as a constant signal, which can either be the start signal or the abort/end signal, depending if it is the input or output of the morphism.

\begin{algorithm}
    \algcap{Ideal resource}{$\S^{QSV}_{\phi, K}$, the ideal quantum state verification resource for $K$ clients and one source.
}
    \label{IR:QSV}
    \begin{algorithmic}
        \INPUT{The clients input $I$.}
        \INPUT{The source inputs $c\in \{0, 1\}$.}
        \IF{$c=0$}
            \STATE $\xi \gets \phi$
        \ELSE
            \STATE $\xi \gets I^{\otimes n}$
        \ENDIF
        \ENSURE The client receives $\xi$.
    \end{algorithmic}
\end{algorithm}
\begin{remark}
    In the category $\CPTP$, morphisms are quantum channels defined on finite-dimensional $C^*$-algebras, or concretely on direct sums of matrix algebras $M_n(\Cc)$. In quantum cryptography, however, we work with density matrices. Density matrices form a subset of all complex matrices, that is $D\left(\bigotimes_{i=1}^k \Cc^{n_i}\right)\subseteq \bigotimes_{i=1}^k M_{n_i}(\Cc)$ for any tuple $(n_1,\dots, n_k)\in\mathbb{N}^k$. Moreover, a quantum channel always maps density operators to density operators. Therefore, all the preliminaries in Section \ref{sec:QITPrelim} also hold for the morphisms in the category $\CPTP$. In the subsequent work we can therefore restrict our analysis to density operators while still working with the morphisms in $\CPTP$.
\end{remark}

\begin{definition}[Quantum state verification]\label{def:QSV}
    Let $\S^{QSV}_{\phi, K}$ be the \emph{quantum state verification} resource for $K$ clients $C=\{i\}_{i=1}^K$, a source $S$ and a target state 
    $\phi \in D\left(\bigotimes_{i=1}^K\Cc^{n_i}\right)$. 
    The source decides with their input $c\in\{0, 1\}$ if the clients receive the target state $\phi$ or the tensor unit $I$. As a morphism, we can type $S^{QSV}_{\phi, K}$ as follows \begin{equation*}
        \S^{QSV}_{\phi, K}: \left(\bigotimes_{i=1}^K I\right) \otimes \left(I\oplus I\right)\rightarrow \left(\left(\bigotimes_{i=1}^K M_{n_i}(\Cc)\right) \oplus \left(\bigotimes_{i=1}^K I\right)\right) \otimes I.
    \end{equation*} We show the ideal resource in \textbf{Ideal resource \ref{IR:QSV}}.
\end{definition}

We further introduce the 1-comb $ t_{\sharp_S}$ that takes the role of a filter, i.e. an operation applied to the ideal resource that shields access that should be only available to dishonest parties. We apply a filter in the honest case to ensure that the source cannot force the ideal resource to abort. Again, we write tensor units explicitly to represent the inputs or outputs of the different parties:
\begin{align*}
    t_{\sharp_S} = \left(I^{\otimes K+1}\to I^{\otimes K} \otimes \begin{pmatrix}
        1 & 0\\
        0 & 0
    \end{pmatrix},\id[\left(\left(\bigotimes_{i=1}^K M_{n_i}(\Cc)\right) \oplus \left(\bigotimes_{i=1}^K I\right)\right) \otimes I]\right).
\end{align*}
We define $\sharp_S = (1\mapsto 1, t_{\sharp_S})$.\\
Now we have all definitions we need to define $\varepsilon$-implementations of quantum state verification.
\begin{definition}[Implementation]\label{def:eps_implementation}
    Let $\bar r$ be any sequence of resources, and $\pi$ be a protocol in form of a morphism in $\ncomb(\V{E}^K\times \CC)$ applicable to $\bar a$. We say $\pi(\bar a)$ is an $\varepsilon$-implementation for quantum state verification if 
    \begin{itemize}
        \item $\pi$ $\varepsilon$-correctly implements $\sharp_S\left(\S^{QSV}_{\phi, K}\right)$ and
        \item $\pi$ $\varepsilon$-securely implements $\S^{QSV}_{\phi, K}$
    \end{itemize}   
    from $\bar r$.
\end{definition}

We still need to list the resources we are considering for the implementation. While we shouldn't be overly restrictive, we must ensure that the clients cannot use these resources alone to fully prepare the target state. Nevertheless, the resource should provide the necessary communication structure. Motivated by this contrast, we describe the abilities provided by the resources and their restrictions below.
\begin{itemize}
    \item $\W$ is a resource allowing the clients to coordinate their verification. It is assumed, that $\W$ either doesn't allow to or is not used to distribute the output state.
    \item $\Q$ is a quantum communication channel from the source to the client in single-client quantum verification.
    \item $\T$ is a quantum communication channel from the source to all clients and allows some quantum communication among the clients. Nevertheless, the graph representing the connectivity of the clients is not connected. We assume $\T$ for quantum state verification. 
    \item $\VV$ is a resource allowing the clients to sample whether they query another state and reveals the decision to the source. We assume once, $\VV$ outputted that no further state should be queried, the parties don't use it or ignore its outputs.
\end{itemize}

%% file: subtexs/04_Results.tex
\section{No-Go result}

\subsection{Simple protocols}
We first consider a simple type of quantum state verification protocols. In this simple setting, an honest source sends $N+1$ copies of the ideal state to the client(s). The client(s) perform a measurement on a random subset of size $N$. If the measurement outcome is $0$ they accept the verification and output the remaining state to the environment. If the measurement outcome is $1$, they output the abort signal $I$ to the environment. 
\begin{definition}[Simple protocol type]\label{def:simple}
    Let $N\geq 0$ be an integer, $\eta\in I^{\oplus N+1}$ a probability distribution, 
    $\phi\in D\left(\bigotimes_{i=1}^K\Cc^{n_i}\right)$ 
    the target state for $K\geq 1$ client(s) and 
    $\mu:\left(\bigotimes_{i=1}^KM_{n_i}(\Cc)\right)^{\otimes N}\rightarrow I\oplus I$ 
    a measurement. $\pi^{SP}$ is defined by the two algorithms $\pi^{SP}_S$ and $\pi^{SP}_C$, where $\pi^{SP}_S$ describes the protocol followed by the source preparing the states and $\pi^{SP}_C$ the protocol followed by the client(s) to verify the states.
    \begin{algorithm}
        \algcap{Protocol}{The protocol $\pi^{SP}$ of the source and the (joint) protocol of the client(s). $N$, $K$ and $\phi$ are publicly known and fixed per protocol instance.}\vspace{1.5em}
        \textbf{Source's protocol $\pi^{SP}_S$:}
        \begin{algorithmic}[1]
            \STATE Prepare $N+1$ copies of the target state, i.e. $I \rightarrow \phi^{\otimes (N+1)}$.
            \STATE Send these copies to the client(s).
        \end{algorithmic}
        \vspace{1em}
        \textbf{Client's protocol $\pi^{SP}_C$:}
        \begin{algorithmic}[1]
            \STATE The client(s) receive their respective share of each of the $N+1$ states in $D\left(\bigotimes_{i=1}^K\Cc^{n_i}\right)$, i.e. $ I\rightarrow \bigotimes_{i=1}^{N+1}\rho_i$.
            \STATE The client(s) sample the output register: $r\gets \eta$
            \IF{$r=1$}
                \STATE $\bigotimes_{i=1}^{N+1}\rho'_i\gets \texttt{MOVE-BACK}_{N+1, 1}(\bigotimes_{i=1}^{N+1}\rho_i)$.
            \ELSIF{$r=...$}
                \STATE $\vdots$
            \ELSE
                \STATE $\bigotimes_{i=1}^{N+1}\rho'_i\gets \texttt{MOVE-BACK}_{N+1, N+1}(\bigotimes_{i=1}^{N+1}\rho_i)$.
            \ENDIF
            \STATE Perform $\texttt{FORGET-BRANCH}_{N+1, \Hi_c}$ and get $\bigotimes_{i=1}^{N+1}\rho''_i$. 
            \STATE Perform the measurement $\mu$ on the first $N$ registers, the result is $s$. The remaining register is now called $\rho'''$.
            \IF{$s=0$}
                \STATE Output $\rho'''$, distributed to the clients.
            \ELSE
                \STATE Output $\Tr(\rho''')$ to each client.
            \ENDIF
        \end{algorithmic}
        \end{algorithm} 
\end{definition}

Where the else-if structure, $\texttt{MOVE-BACK}$, and $\texttt{FORGET-BRANCH}$ are defined in Definitions~\ref{def:elif},~\ref{def:move-back}, and~\ref{def:forget-branch}. With the definition of the simple protocol type, we can derive our first result. We show that there is no composably secure single- or multi-client quantum state verification protocol that is efficient. Intuitively, both players, the client(s) and the distinguisher, are limited by the Holevo-Helstrom theorem. The crucial observation is, that the distinguisher can choose a state that is likely to pass the test of the client(s), but is still non-negligibly far away from the target state. 
We start with single-client state verification, i.e. we consider a single client, who is not able to prepare a state at all. 
For any simple protocol $\pi$ with $N+1$ rounds we show that $\pi$ cannot be an $negl(N)$-implementation as defined in \ref{def:eps_implementation}, \ie $\pi$ cannot be negligible close to the ideal functionality and efficient.

\begin{theorem}[No efficient single-client state verification with fixed number of rounds] \label{thm:STSimple}
    Let $\pi = \{\pi_S, \pi_C\}$ be a simple protocol (see Def.\ref{def:simple}). Then there exists a morphism $\V{A}_{\rho}\in\A(\pi)$ such that for all $\V{B}_{\rho}\in\A\left(\id[\S^{QSV}_{\phi, 1}]\right)$ it holds 
    \begin{align*}
   \hdia{RF(\pi)(\bar a)}{ RF(\sharp_S)\left(\S^{QSV}_{\phi, 1}\right)} + \hdia{RF(\V{A}_{\rho})(\bar a)}{ RF(\V{B}_{\rho})\left(\S^{QSV}_{\phi, 1}\right)}\geq \varepsilon,
    \end{align*}
    where $\bar a = (\Q)^{\times N+1}$ consists of $N+1$ copies on quantum communication from the source to the client, $\varepsilon=\nicefrac{1}{27N}$, if $\phi$ is mixed and $\varepsilon=\nicefrac{1}{8\sqrt{N}}$, if $\phi$ is pure.
\end{theorem}
\begin{proof}
    First, we consider the setting, where the source is honest, i.e.~$\hdia{RF(\pi)(\bar a)}{ RF(\sharp_S)\left(\S^{QSV}_{\phi, 1}\right)}$. Half the diamond distance being the distinguishing advantage, we know that this quantity is lower bounded by the difference of the probabilities that the real and ideal resource output $0$, that is
    \begin{align}\label{eq:simple_hon_measbound}
        \hdia{RF(\pi)(\bar a)}{ RF(\sharp_S)\left(\S^{QSV}_{\phi, 1}\right)} \geq \left|Pr[\left(\pi(\bar a)\right)(I)=0] -
        Pr\left[\left(\sharp_S\left(\S^{QSV}_{\phi, 1}\right)\right)(I)=0\right]\right|.
    \end{align}
    We consider the channel $\M=\Tr_{M_n(\Cc)}\oplus \id[I]$. If we apply $\M$ to the resource's output, we find 
    \begin{align*}
        Pr[\left(\M\circ \pi(\bar a)\right)(I)=0] &= \braket{\mu(0)}{\phi^{\otimes N}}\\
        Pr\left[\left(\M\circ \sharp_S\left(\S^{QSV}_{\phi, 1}\right)\right)(I)=0\right] &= 1.
    \end{align*}
    However the diamond distance is contractive with respect to CPTP maps, i.e. $\M$ cannot increase the diamond distance, which implies
    \begin{align}\label{eq:bound_simple_honest}
        &\hdia{RF(\pi)(\bar a)}{ RF(\sharp_S)\left(\S^{QSV}_{\phi, 1}\right)} \nnn &\geq \hdia{RF(\M\circ\pi)(\bar a)}{RF(\M\circ \sharp_S)\left(\S^{QSV}_{\phi, 1}\right)} \geq 1-\braket{\mu(0)}{\phi^{\otimes N}}.
    \end{align}
    If the source is dishonest, we use the family of attacks $\{\V{A}_{\rho}\}_{\rho\in D(\Cc^n)}$. We apply these attacks to $\Q^{\times N+1}$ and obtain a channel from $I$ to $M_n(\Cc) \oplus I$. 
    $\V{A}_{\rho}$ inputs $N+1$ copies of as state $\rho\in D(\Cc^n)$ and implements $\pi_C$ on the client side.
    To bound $\hdiail{RF(\V{A}_{\rho})(\bar a)}{ RF(\V{B}_{\rho})\left(\S^{QSV}_{\phi, 1}\right)}$ we need to model any attack on $\id[\S^{QSV}_{\phi, 1}]$. For the input part of the comb we consider any channel from $I$ to $I\oplus I$. The output part can only be the identity, since there is no output on the side of the source and the client is honest. Hence, $\V{B}_{\rho}$ can only input a single classical bit $c$ to the ideal resource $\S^{QSV}_{\phi, 1}$. Therefore, we can assume without loss of generality, that there is a probability $p_{N}(\rho)$ for the simulator to input $c=1$, i.e. $p_{N}(\rho) = Pr\left[c=0\mid c\gets \V{B}_{\rho}\left(\S^{QSV}_{\phi, 1}\right)\right]$.
    As in the honest case, we want to apply a measurement channel to the outputs of the channels obtained by $\V{A}_{\rho}(\bar a)$ and $\V{B}_{\rho}\left(\S^{QSV}_{\phi, 1}\right)$. Let $\M_D= \{\gamma(0)\oplus 1, \gamma(1) \oplus 0\}$ be a measurement channel, where $\gamma$ is an arbitrary binary measurement. It follows\footnote{Note, that the attack is i.i.d., hence the probability distribution $\eta$ used by the client(s), does not occur in the analysis.}
    \begin{align*}
        Pr[\M_D\circ\V{A}_{\rho}(\bar a)(I)=1] &= \braket{\gamma(1)}{\rho}\braket{\mu(0)}{\rho^{\otimes N}}\\
        Pr\left[\M_D \circ \V{B}_{\rho}\left(\S^{QSV}_{\phi, 1}\right)(I)=1\right] &= \braket{\gamma(1)}{\phi}p_{N}(\rho).
    \end{align*}
    And the Holevo-Helstrom Theorem implies
    \begin{align}\label{eq:simple_dishon_meas1}
        \hdia{RF(\V{A}_{\rho})(\bar a)}{ RF(\V{B}_{\rho})\left(\S^{QSV}_{\phi, 1}\right)} \geq \left|\braket{\gamma(1)}{\rho}\braket{\mu(0)}{\rho^{\otimes N}} - \braket{\gamma(1)}{\phi}p_{N}(\rho)\right|.
    \end{align}
    When considering $\gamma(1) = \mathbbm{1}_{\Cc^n}$, we find
    \begin{align*}
        \hdia{RF(\V{A}_{\rho})(\bar a)}{ RF(\V{B}_{\rho})\left(\S^{QSV}_{\phi, 1}\right)} \geq |\braket{\mu(0)}{\rho^{\otimes N}}-p_N(\rho)|.
    \end{align*}
    We denote $p_N(\rho) = \braket{\mu(0)}{\rho^{\otimes N}}+\delta(N, \rho)$ and find
    \begin{align*}
        \hdia{RF(\V{A}_{\rho})(\bar a)}{ RF(\V{B}_{\rho})\left(\S^{QSV}_{\phi, 1}\right)} \geq |\delta(N, \rho)|.
    \end{align*}
    For an arbitrary $\gamma(1)$, we then have
    \begin{align*}
        &\hdia{RF(\V{A}_{\rho})(\bar a)}{ RF(\V{B}_{\rho})\left(\S^{QSV}_{\phi, 1}\right)} \nnn\geq & |\braket{\mu(0)}{\rho^{\otimes N}}\braket{\gamma(1)}{\rho}-p_N(\rho)\braket{\gamma(1)}{\phi}|\\ =& |\braket{\mu(0)}{\rho^{\otimes N}}(\braket{\gamma(1)}{\rho}-\braket{\gamma(1)}{\phi}) - \delta(N, \rho)\braket{\gamma(1)}{\phi}|\\
        \geq & |\braket{\mu(0)}{\rho^{\otimes N}}(\braket{\gamma(1)}{\rho}-\braket{\gamma(1)}{\phi})|-|\delta(N, \rho)|\\
        \geq & \braket{\mu(0)}{\rho^{\otimes N}}|\braket{\gamma(1)}{\rho}-\braket{\gamma(1)}{\phi}|-\hdia{RF(\V{A}_{\rho})(\bar a)}{ RF(\V{B}_{\rho})\left(\S^{QSV}_{\phi, 1}\right)}.
    \end{align*}
    If we consider the optimal measurement $\{\gamma(0), \gamma(1)\}$ to distinguish $\rho$ and $\phi$, i.e., the measurement that saturates the Holevo-Helstrom bound, then the above inequality is yields
    \begin{align}\label{eq:bound_simple_dishonest}
        \hdia{RF(\V{A}_{\rho})(\bar a)}{ RF(\V{B}_{\rho})\left(\S^{QSV}_{\phi, 1}\right)}\geq \frac{1}{2}\braket{\mu(0)}{\rho^{\otimes N}}\frac{1}{2}\|\rho-\phi\|_1.
    \end{align}
    Adding both honesty-setting (\ref{eq:bound_simple_honest}) and (\ref{eq:bound_simple_dishonest}), the Holevo-Helstrom bound yields
    \begin{align}
        &\hdia{RF(\pi)(\bar a)}{ RF(\sharp_S)\left(\S^{QSV}_{\phi, 1}\right)} + \hdia{RF(\V{A}_{\rho})(\bar a)}{ RF(\V{B}_{\rho})\left(\S^{QSV}_{\phi, 1}\right)} \nnn
        &\geq \left(1-\braket{\mu(0)}{\phi^{\otimes N}}\right) + \frac{1}{2}\braket{\mu(0)}{\rho^{\otimes N}}\frac{1}{2}\|\rho-\phi\|_1\nnn
        &\geq \frac{1}{4}\|\rho-\phi\|_1 \left(1-\left(\braket{\mu(0)}{\rho^{\otimes N}} - \braket{\mu(0)}{\phi^{\otimes N}}\right)\right)\nnn
        &\geq \frac{1}{4}\|\rho-\phi\|_1 \left(1-\frac{1}{2}\|\rho^{\otimes N}-\phi^{\otimes N}\|_1\right).\label{eq:boundOfDist}
    \end{align}
    Depending on whether target state $\phi$ is pure or mixed we obtain a different bound. We first consider the mixed case. Using Lemma \ref{lemma:Multi_copy_dist}, we find
    \begin{align*}
        &\hdia{RF(\pi)(\bar a)}{ RF(\sharp_S)\left(\S^{QSV}_{\phi, 1}\right)} + \hdia{RF(\V{A}_{\rho})(\bar a)}{ RF(\V{B}_{\rho})\left(\S^{QSV}_{\phi, 1}\right)}\\
        &\geq \frac{1}{4}\|\rho-\phi\|_1 \left(1-\sqrt{1-(1-\|\rho-\phi\|_1)^N}\right)
    \end{align*}
    Now we fix $\V{A}_{\rho}$ such that $\frac{1}{2}\|\rho-\phi\|_1 = \nicefrac{\alpha}{N}$, which give us
    \begin{align*}
        &\hdia{RF(\pi)(\bar a)}{ RF(\sharp_S)\left(\S^{QSV}_{\phi, 1}\right)} + \hdia{RF(\V{A}_{\rho})(\bar a)}{ RF(\V{B}_{\rho})\left(\S^{QSV}_{\phi, 1}\right)}\\
        &\geq \frac{\alpha}{2N} \left(1-\sqrt{1-\left(1-\frac{2\alpha}{N}\right)^N}\right)\\
        &\geq \frac{\alpha}{2N} \left(1-\sqrt{2\alpha}\right),
    \end{align*}
    where we used $(1-\frac{\beta}{k})^{k}\geq 1-\beta$ for $k\in\mathbb{N}$ and $|\beta|\leq k$.
    This is maximized for $\alpha=\nicefrac{2}{9}$ which gives
    \begin{align}\label{eq:fin_bound_mixed}
        \hdia{RF(\pi)(\bar a)}{ RF(\sharp_S)\left(\S^{QSV}_{\phi, 1}\right)} + \hdia{RF(\V{A}_{\rho})(\bar a)}{ RF(\V{B}_{\rho})\left(\S^{QSV}_{\phi, 1}\right)} \geq \frac{1}{27N}.
    \end{align}
    Next, we consider $\phi = \ketbra{\phi}$ to be a pure state, and we can choose a pure state $\rho=\ketbra{\psi}$ as well. Using (\ref{eq:pureState}) plugged into (\ref{eq:boundOfDist}) we obtain
    \begin{align*}
        &\hdia{RF(\pi)(\bar a)}{ RF(\sharp_S)\left(\S^{QSV}_{\phi, 1}\right)} + \hdia{RF(\V{A}_{\rho})(\bar a)}{ RF(\V{B}_{\rho})\left(\S^{QSV}_{\phi, 1}\right)}\\
        &\geq \frac{1}{2} \sqrt{1-|\braket{\psi}{\phi}|^2} \left(1-\sqrt{1-|\braket{\psi}{\phi}|^{2N}}\right).
    \end{align*}
    Replacing $\sqrt{1-|\braket{\psi}{\phi}|^2}$ with $\tau$ yields
    \begin{align*}
        &\hdia{RF(\pi)(\bar a)}{ RF(\sharp_S)\left(\S^{QSV}_{\phi, 1}\right)} + \hdia{RF(\V{A}_{\rho})(\bar a)}{ RF(\V{B}_{\rho})\left(\S^{QSV}_{\phi, 1}\right)}\\
        &\geq \frac{\tau}{2} \left(1-\sqrt{1-(1-\tau^2)^N}\right).
    \end{align*}
    Now we choose $\ket{\psi}$ such that $\tau=\nicefrac{1}{2\sqrt{N}}$ and we find again using $(1-\frac{\beta}{k})^{k}\geq 1-\beta$ for $k\in\mathbb{N}$ and $|\beta|\leq k$
    \begin{align}\label{eq:fin_bound_pure}
        &\hdia{RF(\pi)(\bar a)}{ RF(\sharp_S)\left(\S^{QSV}_{\phi, 1}\right)} + \hdia{RF(\V{A}_{\rho})(\bar a)}{ RF(\V{B}_{\rho})\left(\S^{QSV}_{\phi, 1}\right)}\geq \frac{1}{8\sqrt{N}}.
    \end{align}\qed
\end{proof}
We can also extend the result for simple protocols for multi-client quantum state verification. 

\begin{theorem}[No efficient secure state verification with fixed number of rounds for entangled states]\label{thm:QSVsimple}
    Let $\pi = \{\pi_S\}\cup\{\pi_i\}_{i\in C}$ be a simple protocol as defined in Def. \ref{def:simple}. Then there exists a morphism $\V{A}_{\rho}\in\A(\pi)$ such that for all $\V{B}_{\rho}\in\A\left(\id[\S^{QSV}_{\phi, K}]\right)$ it holds 
    \begin{align*}
       \hdia{\pi(\bar a)} {\sharp_S\left(\S^{QSV}_{\phi, K}\right)} + \hdia{\V{A}_{\rho}(\bar a)}{ \V{B}_{\rho}\left(\S^{QSV}_{\phi, K}\right)}\geq \varepsilon,
    \end{align*}
    where $\bar a = (\T^{\times N+1}, \W)$ consits of $N+1$ copies of $\T$ used to distribute the states and one copy of $\W$ used to coordinate the verification task. Further, we find $\varepsilon=\nicefrac{1}{27N}$, if $\phi$ is mixed and $\varepsilon=\nicefrac{1}{8\sqrt{N}}$, if $\phi$ is pure.
\end{theorem}
The proof of Theorem \ref{thm:QSVsimple} is just a repetition of the proof of Theorem \ref{thm:STSimple} -- having multiple clients does not change the fundamental inequalities we used.

\subsection{General protocols}\label{section:general_protocol}
Next, we consider general protocols of quantum state verification. In this setting, the client(s) sample $r$ and $i$, where $r$ is the number of states to be sent by the source and $i$, the number of those states that are then measured to verify the state. To this end the clients use a joint probability distribution $p(r,i)$. Again, if the measurement outcome is $0$, they accept the verification and output one of the remaining $r-i$ states to the environment. If the measurement outcome is $1$, they output the abort signal $I$ to the environment.

For the general protocol, we need to adapt our categorical model for the multi-client case. For the single-client case, the client is still not able to create a state and needs to output a state created by the distinguisher. This translates to $p(l,l)=0$ for all $l$. In the multi-client case, the clients can prepare a state. However, the state they can prepare is separable with respect to a partition with respect to which to the target state is entangled. Otherwise, the clients would not require an external source in the first place.

In the proof of Theorems \ref{thm:general_multi} and \ref{thm:general_single} we allow for an arbitrary number of verification rounds, i.e. $r$ and $i$ are arbitrary positive integers. However, in our categorical modeling, we set an upper bound on the number of rounds. The bound can be arbitrarily chosen and the two settings are then equivalent. Indeed let $D$ be the upper bound, then we can set $p(r,i)=0$ for all $r\leq D$. On the other hand, the upper bound $D$ can be chosen arbitrarily large and is therefore no significant restriction. The probabilities for $r>D$ can be chosen arbitrarily small, but non-zero, for every distribution.\\[1em]
With this general formulation of protocols, we find that any implementation of quantum state verification either has a non-negligible distance to the ideal functionality according to Def. \ref{def:eps_implementation} or is inefficient in the number of rounds.

\begin{theorem}[No efficient secure state verification for entangled states]\label{thm:general_multi}
    Let $\pi=\{\pi_S\}\cup\{\pi_i\}_{i\in C}$ be any protocol applied on resources $\bar a = (\T^{\times M},\VV^{\times M}, \W)$, and $D$ is an upper bound for the number of rounds. $\pi(\bar a)$ is an implementation with the following properties
    \begin{itemize}
        \item $K$ clients sample the number of rounds $r$ and the number of verification rounds $0\leq i\leq r$ from any joint distribution $p(r, i)$,
        \item the clients perform a measurement $\mu_{r, i}$ and accept the outcome if the result is $0$,
        \item if $r=i$, the clients prepare a state $\chi\in D\left(\bigotimes_{j=1}^K \Cc^{n_j}\right)$, where $\chi$ and the target state $\phi$ are separable respectively entangled with respect to a particular partition,
        \item  if $r\neq i$, the clients output one of the $r-i$ states drawn to any unspecified distribution.
    \end{itemize}
    It then follows that there exists $\V{A}_{\rho}\in\A(\pi)$ such that for all $\V{B}_{\rho}\in\A\left(\id[\S^{QSV}_{\phi, K}]\right)$ it holds that
    \begin{align*}
  \hdia{\pi(\bar a)}{\sharp_S\left(\S^{QSV}_{\phi, K}\right)} + \hdia{\V{A}_{\rho}(\bar a)}{ \V{B}_{\rho}\left(\S^{QSV}_{\phi, K}\right)}\geq \varepsilon,
    \end{align*}
    where $\varepsilon=\nicefrac{1}{27N}$, if $\phi$ is mixed and $\varepsilon=\nicefrac{1}{8\sqrt{N}}$, if $\phi$ is pure and $N$ is the expected number of rounds.
\end{theorem}
\begin{proof}[sketch]
    The proof is very similar to the proof of Thm. \ref{thm:STSimple}. Hence, we only sketch the proof idea and provide the formal proof in the appendix \ref{app2}.
    \begin{enumerate}
    \setlength\itemsep{1em}
        \item If the source is honest, we use the same measurement $\M$ as in (\ref{eq:simple_hon_measbound}) to find a lower bound for $\hdiail{\pi(\bar a)}{\sharp_S\left(\S^{QSV}_{\phi, K}\right)}$. 
        \item If the source is dishonest, we use a similar i.i.d. attack and the measurement $\M_D$ used for (\ref{eq:simple_dishon_meas1}), which yields a lower bound $\hdiail{\V{A}_{\rho}(\bar a)}{ \V{B}_{\rho}\left(\S^{QSV}_{\phi, K}\right)}$.
        \item Next, we consider $\gamma(1) = \mathbbm{1}_{\bigotimes_{i=1}^K\Cc^{n_i}}$, find a $\delta$ which we use to eliminate the acceptance probability of $\V{B}_{\rho}$.
        \item An argument based on  the direction of $\rho$ and $\gamma$ shows that $\min(\braket{\gamma(1)}{\chi}-\braket{\gamma(1)}{\phi}, \braket{\gamma(1)}{\rho}-\braket{\gamma(1)}{\phi})\geq 0$. We then bring both honesty configurations together.
        \item At last, we use Jensen's inequality, and the fact that $\lambda \geq \nicefrac{1}{2\sqrt{N}}\geq \nicefrac{2}{9N}$ in the asymptotic limit, as $\nicefrac{1}{2}\|\chi-\phi\|_1>0$ is constant, to deduce the same lower bounds as in (\ref{eq:fin_bound_mixed}) and  (\ref{eq:fin_bound_pure}).
    \end{enumerate}
\end{proof}
In fact, the same proof works for single-client quantum state verification as well. The difference is now that $p(l, l)=0$ for every $l\geq 0$ as the client is not able to prepare states at all. Hence, we provide a bound for single-client quantum state verification without restating the proof.
\begin{theorem}[No efficient secure single-client quantum state verification]\label{thm:general_single}
Let $\pi=\{\pi_C, \pi_S\}$ be any protocol applied to resources $\bar a = (\Q^{\times M}, \VV^{\times M}, \W)$ and $D$ be an upper bound on the number of rounds. $\pi(\bar a)$ is an implementation with the following properties
    \begin{itemize}
        \item the client samples the number of rounds $r$ and the number of verification rounds $0\leq i< r$ from any joint distribution $p(r,i)$,
        \item the client performs a measurement $\mu_{r, i}$ and accepts the outcome if the result is $0$,
        \item the client outputs one of the $r-i$ states drawn to any unspecified distribution.
    \end{itemize}
    It then follows that there exists $\V{A}_{\rho}\in\A(\pi)$ such that for $\V{B}_{\rho}\in\A\left(\id[\S^{QSV}_{\phi, 1}]\right)$ it holds
    \begin{align*}
    \hdia{\pi(\bar a)}{ \sharp_S\left(\S^{QSV}_{\phi, K}\right)} + \hdia{\V{A}_{\rho}}{ \V{B}_{\rho}\left(\S^{QSV}_{\phi, 1}\right)}\geq \varepsilon,
    \end{align*}
    where $\varepsilon=\nicefrac{1}{27N}$, if $\phi$ is mixed and $\varepsilon=\nicefrac{1}{8\sqrt{N}}$, if $\phi$ is pure and $N$ is the expected number of rounds.
\end{theorem}

\begin{remark}[Categorical modeling]
Our categorical modeling is more restrictive than necessary. Indeed, in the single-client case, the no-state-preparation assumption is overly restrictive. It is enough to assume that the client cannot prepare states that are too close to the target state, otherwise a source would be superfluous. A reasonable assumption could be that the client can only prepare states that are outside a finite-size ball around the target state. Nevertheless, the categorical model remains simpler and more intuitive with stronger restrictions.

In the multi-client case, we only need to ensure that the target state is entangled with respect to a partition for which the clients cannot generate entanglement. This holds as long as there is a subset of clients which is not connected via quantum channels. If the clients are not connected via quantum channels they can only perform separable operations with respect to that particular partition. Since separable operations cannot create entanglement, we can be certain that any state the clients prepare will be far enough from any target state that is entangled with respect to the same partition. This restriction is sensible, as otherwise, an external source is superfluous.
\end{remark}

%% file: subtexs/05_Discussion.tex
\section{Discussion}
\subsection{Summary}
In our work, we first present how to use the categorical composable cryptography framework for quantum cryptography. For that, we introduce a resource theory based on n-combs on \CPTP. The instantiation of the framework we presented can contribute to a deeper understanding of composable quantum cryptography as it defines protocols, resources and attacks rigorously while still being applicable without additional effort.\\
Using this instantiation of the framework, we prove that quantum state verification can not be efficient and secure if one relies on the usual cut-and-choose technique, i.e. uses one of the rounds directly as output. Indeed we show that in the usual cut-and-chose regime a quantum state verification protocol is either to far from the ideal quantum state verification resource and therefore insecure or it is inefficient in the number of rounds. Our result is agnostic about the target state, the number of clients, and used resources, except for a few reasonable restrictions. These restrictions should only prevent the clients from preparing the target state themselves and are the motivation to use such a protocol in the first place. 
Although we only consider quantum state verification for our results, one finds direct implications for other primitives. One example is self-testing, in which a party prepares states with an untrusted device and measures them with a different untrusted device to verify the preparation. It is easy to see that this is even harder than quantum state verification as the measurement device is not trustworthy in self-testing, i.e. our result extends naturally to self-testing.\\
\subsection{Discussion of the assumptions}
The strength and generality of our results stems from the fact that we only use very few and simple assumptions. A fundamental assumption is the inability of the clients to prepare the state themselves. While we argue that this assumption comes naturally in the setting of quantum state verification, we modeled the categorical representation of this assumption stricter than necessary. For the proof to work, we only need that the client can not prepare states that are far enough from the target state, especially since we consider the asymptotic behavior. However, this restriction is complicated to model in a category for the client, which motivates the stricter modeling. We leave it open to future work to find less restrictive categories which implement the assumption.\\[1em] 
Another assumption is that the clients output the state as received. While this seems to be the natural approach for verification, our work shows that it fails. In fact, questioning this assumption might lead to a workaround, which we discuss in more detail in the next section.\\[1em] 
At last, one might be tempted to see the framework we used as an assumption. Because of that, we emphasize that one can find the same lower bounds for implementations of quantum state verification in other composability frameworks\cite{Unruh10,MauRen11,Mau11}. This fact is already reflected by how we present the proof. Indeed, measuring the output and input choice would be part of a distinguisher in other frameworks, such as abstract cryptography. The simulator would implement the attack $\mathbf{B}_{\rho}$ on the ideal resource but it would have to obey the same restrictions regarding the input of the ideal resource as the attack. In the end, the inequalities are the same.
However, the explicit and strict typing of the categorical composable cryptography framework allows for rigorous proofs without significant overhead once the user understands the framework. Further, the flexibility of choosing the appropriate attack model enables the user to analyse more restricted or complicated adversarial situations such as honest-but-curios or non-colluding adversaries. Our proof provides an example of this flexibility: One could restrict the attack model to i.i.d. attacks and still find the same result.
\subsection{Possible workarounds and open questions}
The lower bounds we presented are an inherent property of quantum state verification in a cut-and-choose fashion. They raise the question of how to circumvent this lower bound and what consequences follow. First, one should recall the implication of the result: One can not use quantum state verification in a modular manner for cryptography as one can not have efficiency and security. This no-go result holds not only for composable cryptography; in the recent work \cite{colisson2024graph}, the authors show that stand-alone secure protocols for quantum state verification are composable secure, where the $\varepsilon$ for composable security is a polynomial of the one for the stand-alone security definition. This lifting implies that our result extends to stand-alone security as well. Either way, the implication is about the modular use or as a protocol for its own sake. However, most times verification is used in the context of a larger protocol, which raises the question of what happens in a non-modular setting. We investigated this question to some extent by post-composing the ideal resource with different kinds of channels on the client side. We found that the lower bounds similarly extend to post-composition with unital channels and measurements in a basis. We present this result in more detail in the Appendix \ref{app1}. 
The idea of post-composition also leads to other approaches. One of these is error-detection: If the server has to prepare the target state in an error-detection code and the clients run the verification on the encoding, they can decode the output and eliminate or detect errors introduced by a dishonest server. Similar techniques are already used in verifiable delegated quantum computing \cite{Kashefi_2017}, which already indicates that other primitives using verification in quantum cryptography could be affected by similar lower bounds. 
So, while we show that the naive approach to quantum state verification is doomed to fail, many open questions remain, and possible workarounds may exist.

%% file: subtexs/a01_Appendix.tex
\section{Appendix}
\subsection{Detailed proof for general protocols}\label{app2}
Before we restate the theorem, we note that we'll encounter expectation values of function of the following type in the proof, for $0 < a< 1$
    \begin{align}
        f_{a}(X) = \sqrt{1-a^{X}}.
    \end{align}
The functions $f_a$ are concave, since their second derivative is negative, indeed 
\begin{align}
    f''_a(Z) = \frac{d^2f}{dX^2}(Z) = -\frac{a^{Z} \ln\left(a\right)^{2} \left(2-a^{Z}\right)}{4 \left(\sqrt{1 - a^{Z}}\right)^{3}}<0.
\end{align}
 Since $f_a$ is concave, we can use Jensen's inequality for the expectation value of $f_{a}$ and we find
    \begin{align}
        E(f_{a}(X))\leq f_{a}(E(X)).
    \end{align}
\begin{theorem}[No efficient secure state verification for entangled states]\label{thm:general_multiApp}
    Let $\pi=\{\pi_S\}\cup\{\pi_i\}_{i\in C}$ be any protocol applied on resources $\bar a = (\T^{\times M},\VV^{\times M}, \W)$, and $D$ is an upper bound for the number of rounds. $\pi(\bar a)$ is an implementation with the following properties
    \begin{itemize}
        \item $K$ clients sample the number of rounds $r$ and the number of verification rounds $0\leq i\leq r$ from any joint distribution $p(r, i)$,
        \item the clients perform a measurement $\mu_{r, i}$ and accept the outcome if the result is $0$,
        \item if $r=i$, the clients prepare a state $\chi\in D\left(\bigotimes_{j=1}^K \Cc^{n_i}\right)$, where $\chi$ and the target state $\phi$ are separable respectively entangled with respect to a particular partition,
        \item  if $r\neq i$, the clients output one of the $r-i$ states drawn to any unspecified distribution.
    \end{itemize}
    It then follows that there exists $\V{A}_{\rho}\in\A(\pi)$ such that for all $\V{B}_{\rho}\in\A\left(\id[\S^{QSV}_{\phi, K}]\right)$ it holds that
    \begin{align}
  \hdia{RF(\pi)(\bar a)}{RF(\sharp_S)\left(\S^{QSV}_{\phi, K}\right)} + \hdia{RF(\V{A}_{\rho})(\bar a)}{ RF(\V{B}_{\rho})\left(\S^{QSV}_{\phi, K}\right)}\geq \varepsilon,
    \end{align}
    where $\varepsilon=\nicefrac{1}{27N}$, if $\phi$ is mixed and $\varepsilon=\nicefrac{1}{8\sqrt{N}}$, if $\phi$ is pure and $N$ is the expected number of rounds.
\end{theorem}
\begin{proof}
    First we consider correctness, i.e. assume the source is honest. Again, we can bound the diamond-distance by composing with a $\M = Tr_{\bigotimes_{i\in C}M_{n_i}(\Cc)}\otimes \id[I]$ and find
    \begin{align}
        \hdia{RF(\pi)(\bar a)}{RF(\sharp_S)\left(\S^{QSV}_{\phi, K}\right)} \geq 1 - \sum_{r=0}^{\infty}\sum_{i=0}^rp(r, i)\braket{\mu_{r, i}(0)}{\phi^{\otimes i}}. 
    \end{align}
    If the source is dishonest, we consider a family of attacks $\{\V{A}_{\rho}\}_{\rho\in D\left(\bigotimes_{i\in C}\Cc^{n_i}\right)}$ which prepares and inputs for every query the state $\rho$ on the source's side and implements $\pi_C$ as the clients are considered to be honest. As $\dom(\V{A}_{\rho}(\bar a)) = I$, the domain of any suitable attack $\V{B}_{\rho}\in \A\left(\id[\S^{QSV}_{\phi, K}]\right)$ must have the same domain, i.e. prepares and inputs a binary distribution $\{q(\rho),1-q(\rho)\}$, inputs this at the source's interface and acts as the identity on the client's side. 
    Again, with $\M_D= \{\gamma(0)\oplus 1, \gamma(1) \oplus 0\}$ we find
    \begin{align}
        &Pr[\M_D\circ \V{A}_{\rho}(\bar a)(I)=1] \nnn&= \braket{\gamma(1)}{\rho}\sum_{r=1}^{\infty}\sum_{i=0}^{r-1}p(r, i)\braket{\mu_{r, i}(0)}{\rho^{\otimes i}} + \braket{\gamma(1)}{\chi}\sum_{r=0}^{\infty}p(r, r)\braket{\mu_{r, r}(0)}{\rho^{\otimes r}}\\
        &Pr[\M_D\circ \V{B}_{\rho}\left(\S^{QSV}_{\phi, K}\right)(I)=1] = \braket{\gamma(1)}{\phi} q(\rho)
    \end{align}
    With $\braket{\gamma(1)}{\phi} = \braket{\gamma(1)}{\rho} = \braket{\gamma(1)}{\chi}=1$ we find:
    \begin{align}
        \hdia{RF(\V{A}_{\rho})(\bar a)}{ RF(\V{B}_{\rho})\left(\S^{QSV}_{\phi, K}\right)} \geq& \left|\left(\sum_{r=0}^{\infty}\sum_{i=0}^rp(r, i)\braket{\mu_{r, i}(0)}{\rho^{\otimes i}}\right)-q(\rho)\right| = |\delta|,
    \end{align}
    with
    \begin{align}
        \delta &= \left(\sum_{r=0}^{\infty}\sum_{i=0}^rp(r, i)\braket{\mu_{r, i}(0)}{\rho^{\otimes i}}\right) - q(\rho).
    \end{align}
    With that, we find for any measurement $\gamma$
    \begin{align}
        \hdia{RF(\V{A}_{\rho})(\bar a)}{ RF(\V{B}_{\rho})\left(\S^{QSV}_{\phi, K}\right)} &\geq \left|\braket{\gamma(1)}{\rho}\sum_{r=1}^{\infty}\sum_{i=0}^{r-1}p(r, i)\braket{\mu_{r, i}(0)}{\rho^{\otimes i}} \right.\nnn
        &+\left. \braket{\gamma(1)}{\chi}\sum_{r=0}^{\infty}p(r, r)\braket{\mu_{r, r}(0)}{\rho^{\otimes r}}-\braket{\gamma(1)}{\phi} q(\rho) \right|\nnn
        &= \left|\left(\braket{\gamma(1)}{\rho}-\braket{\gamma(1)}{\phi}\right)\sum_{r=1}^{\infty}\sum_{i=0}^{r-1}p(r, i)\braket{\mu_{r, i}(0)}{\rho^{\otimes i}} \right.\nnn&+\left. \left(\braket{\gamma(1)}{\chi}-\braket{\gamma(1)}{\phi}\right)\sum_{r=0}^{\infty}p(r, r)\braket{\mu_{r, r}(0)}{\rho^{\otimes r}}+\braket{\gamma(1)}{\phi}(\delta)\right|\nnn
        &\geq \left|\left(\braket{\gamma(1)}{\rho}-\braket{\gamma(1)}{\phi}\right)\sum_{r=1}^{\infty}\sum_{i=0}^{r-1}p(r, i)\braket{\mu_{r, i}(0)}{\rho^{\otimes i}} \right.\nnn&+\left. \left(\braket{\gamma(1)}{\chi}-\braket{\gamma(1)}{\phi}\right)\sum_{r=0}^{\infty}p(r, r)\braket{\mu_{r, r}(0)}{\rho^{\otimes r}}\right|-\left|\delta\right|,\end{align}
        which implies
        \begin{align}
        &\hdia{RF(\V{A}_{\rho})(\bar a)}{ RF(\V{B}_{\rho})\left(\S^{QSV}_{\phi, K}\right)} \geq \frac{1}{2}\left|\left(\braket{\gamma(1)}{\rho}-\braket{\gamma(1)}{\phi}\right)\sum_{r=1}^{\infty}\sum_{i=0}^{r-1}p(r, i)\braket{\mu_{r, i}(0)}{\rho^{\otimes i}} \right.\nnn&+\left. \left(\braket{\gamma(1)}{\chi}-\braket{\gamma(1)}{\phi}\right)\sum_{r=0}^{\infty}p(r, r)\braket{\mu_{r, r}(0)}{\rho^{\otimes r}}\right|
    \end{align}
    We can choose the direction of $\rho$ and $\gamma(1)$ such that $\braket{\gamma(1)}{\chi}\geq \braket{\gamma(1)}{\phi} \leq \braket{\gamma(1)}{\rho}$ and define $\lambda = \min(\braket{\gamma(1)}{\chi}-\braket{\gamma(1)}{\phi}, \braket{\gamma(1)}{\rho}-\braket{\gamma(1)}{\phi})$:
    \begin{align}
        \hdia{RF(\V{A}_{\rho})(\bar a)}{ RF(\V{B}_{\rho})\left(\S^{QSV}_{\phi, K}\right)}&\geq \frac{\lambda}{2}\left|\sum_{r=1}^{\infty}\sum_{i=0}^{r-1}p(r, i)\braket{\mu_{r, i}(0)}{\rho^{\otimes i}} + \sum_{r=0}^{\infty}p(r, r)\braket{\mu_{r, r}(0)}{\rho^{\otimes r}}\right|\nnn
        & = \frac{\lambda}{2}\left|\sum_{r=0}^{\infty}\sum_{i=0}^rp(r, i)\braket{\mu_{r, i}(0)}{\rho^{\otimes i}}\right|.
    \end{align}
    Now, we again consider both honesty configurations together
    \begin{align}
        &\hdia{RF(\V{A}_{\rho})(\bar a)}{ RF(\V{B}_{\rho})\left(\S^{QSV}_{\phi, K}\right)}+ \hdia{RF(\pi)(\bar a)}{RF(\sharp_S)\left(\S^{QSV}_{\phi, K}\right)}\nnn&\geq \frac{\lambda}{2}\left|\sum_{r=0}^{\infty}\sum_{i=0}^rp(r, i)\braket{\mu_{r, i}(0)}{\rho^{\otimes i}}\right|+\left|1 - \sum_{r=0}^{\infty}\sum_{i=0}^rp(r, i)\braket{\mu_{r, i}(0)}{\phi^{\otimes i}}\right|\nnn
        &\geq \frac{\lambda}{2}\left|1+\sum_{r=0}^{\infty}\sum_{i=0}^rp(r, i)\left(\braket{\mu_{r, i}(0)}{\rho^{\otimes i}} - \braket{\mu_{r, i}(0)}{\phi^{\otimes i}}\right)\right|\nnn
        &\geq \frac{\lambda}{2}\left(1-\sum_{r=0}^{\infty}\sum_{i=0}^rp(r, i)\left|\braket{\mu_{r, i}(0)}{\rho^{\otimes i}} - \braket{\mu_{r, i}(0)}{\phi^{\otimes i}}\right|\right)\nnn
        &\geq \frac{\lambda}{2}\left(1-\sum_{r=0}^{\infty}\sum_{i=0}^rp(r, i)\|\rho^{\otimes i}-\phi^{\otimes i}\|_1\right).
    \end{align} 
    For mixed states we find with $\nicefrac{1}{2}\|\rho-\phi\|_1=\nicefrac{2}{9N}$
    \begin{align}
        &\hdia{RF(\V{A}_{\rho})(\bar a)}{ RF(\V{B}_{\rho})\left(\S^{QSV}_{\phi, K}\right)}+ \hdia{RF(\pi)(\bar a)}{RF(\sharp_S)\left(\S^{QSV}_{\phi, K}\right)}\nnn &\geq\frac{\lambda}{2}\left(1-\sum_{r=0}^{\infty}\sum_{i=0}^rp(r, i)\sqrt{1-(1-\nicefrac{4}{9N})^{i}}\right)\nnn
        &\geq \frac{\lambda}{2}\left(1-\sum_{r=0}^{\infty}\sqrt{1-(1-\nicefrac{4}{9N})^r}\left(\sum_{i=0}^rp(r, i)\right)\right)\geq \frac{\lambda}{2}\left(1-\sqrt{1-(1-\nicefrac{4}{9N})^{N}}\right),
    \end{align}
    where we used Jensen's inequality with $N$ being the average number of rounds. If $N$ is large enough and the clients are not able to prepare $\phi$, we find the same expression as in the proof of Thm. \ref{thm:STSimple} and find $\hdiail{\V{A}_{\rho}(\bar a)}{ \V{B}_{\rho}\left(\S^{QSV}_{\phi, K}\right)}+ \hdiail{\pi(\bar a)}{\sharp_S\left(\S^{QSV}_{\phi, K}\right)}\geq\nicefrac{1}{27N}$.\\[1em]
    For pure states, we use again a pure state for $\rho$, $\rho=\ketbra{\psi}$ and find if $N$ is large enough:
     \begin{align}
          &\hdia{RF(\V{A}_{\rho})(\bar a)}{ RF(\V{B}_{\rho})\left(\S^{QSV}_{\phi, K}\right)}+ \hdia{RF(\pi)(\bar a)}{RF(\sharp_S)\left(\S^{QSV}_{\phi, K}\right)}\nnn
         &\geq\frac{\sqrt{1-|\braket{\psi}{\phi}|^2}}{2}\left(1-\sum_{r=0}^{\infty}\sum_{i=0}^rp(r, i)\sqrt{1-\sqrt{1-|\braket{\psi}{\phi}|^{2i}}}\right)\nnn
        &\geq\frac{\sqrt{1-|\braket{\psi}{\phi}|^2}}{2}\left(1-\sum_{r=0}^{\infty}\sqrt{1-\sqrt{1-|\braket{\psi}{\phi}|^{2r}}}\sum_{i=0}^rp(r, i)\right)\nnn&\geq \frac{\sqrt{1-|\braket{\psi}{\phi}|^2}}{2}\left(1-\sqrt{1-|\braket{\psi}{\phi}|^{2N}}\right), \label{eq:mostGen_pure}
    \end{align}
    i.e. we find again $\hdia{RF(\V{A}_{\rho})(\bar a)}{ RF(\V{B}_{\rho})\left(\S^{QSV}_{\phi, K}\right)}+ \hdia{RF(\pi)(\bar a)}{RF(\sharp_S)\left(\S^{QSV}_{\phi, K}\right)}\geq\nicefrac{1}{8\sqrt{N}}$\qed
\end{proof}
\subsection{Extending to post-composition with channels}\label{app1}
We also need to consider the situation where the clients try to overcome or at least improve the no-go result. The clients could do so by applying a channel either before or after the verification.
Since for any channel $\phi$ it holds that
    \begin{align}
        \frac{1}{2}\|\rho-\phi\|_1\geq \frac{1}{2}\|\Lambda(\rho)-\Lambda(\phi)\|_1,
    \end{align}
 precomposing with a channel is of no help, as it could decrease the clients' chance of catching the source cheating. However, applying a channel after accepting the state could yield a good implementation as it cannot increase the distinguishing advantage. For a channel $\Lambda$ performed on the output of $\VV^f_{\phi}$ in the case of no abort, we denote the resulting (ideal) resource $\Lambda\circ\VV^f_{\phi}$. The question is then whether there are channels that can either improve the lower bound or avoid the no-go result all together.
 
Different channels lead to vastly different results. In fact, the analysis in the previous section breaks down completely for some channels. For example, consider a replacement channel:
    \begin{align}
        \Lambda^{\chi}_{repl}(\rho) = Tr(\rho)\chi.
    \end{align}
    We find that $\frac{1}{2}\|\Lambda^{\chi}_{repl}(\rho)-\Lambda^{\chi}_{repl}(\phi)\|_1 =\frac{1}{2}\|\chi-\chi\|_1=0$, i.e. there is no chance a distinguisher could distinguish the implementation and the ideal resource. However, replacement channels are not interesting because they imply that the clients were able to prepare the desired state in the first place, making the source and therefore the verification obsolete.
    
    Since the trace distance is unitarily invariant, the no-go result is upheld under post-composition with unitary channels. In the following we take a closer look at measurement and unital channels.

\subsubsection{Measurement channels}
    We consider the scenario where the clients, after accepting the verification, measure the state and output the outcome. Let the same happen in the ideal setting. Can this be composably secure with negligible distinguishing advantage?     
    The verification works as in the general case. We just need to specify how the distinguisher could distinguish the outputs. We use the general setting described in section~\ref{section:general_protocol}, but with $p(i,i)=0$ for all $i$. Let $d>2$ be the dimension of the output space, then we denote the measurement channel as 
    \begin{align}
        \M(\rho) = \sum_{j=1}^d\ketbra{j}{\xi_j}\rho\ketbra{\xi_j}{j}, 
    \end{align}
    The distinguisher now fixes one $\Tilde{j}$ and only outputs $1$ if the measurement outcome is $\Tilde{j}$.
   We restrict our analysis to pure states. For a large enough $N$, we have that there is an attack $\V{A}_{\psi}\in\A(\pi')$ such that for every attack $\V{B}_{\psi}\in\A(\id[\S^{QSV}_{\phi, K}])$ holds (\ref{eq:mostGen_pure})
\begin{align}\label{eq:lower_bound_measurement}
     &\hdia{RF(\V{A}_{\psi})(\bar a)}{RF(\V{B}_{\psi})\left(\S^{QSV}_{\phi, K}\right)}+ \hdia{\pi'(\bar a)}{\sharp_S\left(\S^{QSV}_{\phi, K}\right)}\nnn &\geq\frac{\sqrt{1-|\braket{\psi}{\phi}|^2}}{2}\left(1-\sqrt{1-|\braket{\psi}{\phi}|^{2N}}\right).
=\end{align}
For our current setting the protocol $\pi$ then includes the measurement and if we set $\ket{\xi}=\ket{\xi_{\Tilde{j}}}$, we find that there is an attack $\V{A}_{\psi}\in\A(\pi)$ such that for every attack $\V{B}_{\psi}\in\A(\id[\M\circ \S^{QSV}_{\phi, K}])$
    \begin{align}
        &\hdia{RF(\V{A}_{\psi})(\bar a)}{RF(\V{B}_{\psi})\left(\M\circ\S^{QSV}_{\phi, K}\right)}+ \hdia{RF(\pi)(\bar a)}{RF(\sharp_S)\left(\M \circ\S^{QSV}_{\phi, K}\right)}\nnn &\geq\frac{\left||\braket{\xi}{\psi}|^2-|\braket{\xi}{\phi}|^2\right|}{2}\left(1-\sqrt{1-|\braket{\psi}{\phi}|^{2N}}\right).\label{eq:noBasisMeas}
    \end{align}
 We choose $\ket{\xi}$ such that $\nicefrac{1}{\sqrt{2}}>\nicefrac{1}{\sqrt{d}}\geq |\braket{\xi}{\phi}|\geq 0$. This is possible since we assume $d\geq3$. 
 We now choose $\nicefrac{\pi}{2}>\eta\geq0$ and $\nicefrac{\pi}{2}>\theta\geq0$ such that $\cos(\theta) = |\braket{\xi}{\phi}|$ and $\cos(\eta)=|\braket{\psi}{\phi}|$. Next, we fix $\ket{\psi}$ to be in the plane spanned by $\ket{\phi}$ and $\ket{\xi}$. We define the following basis
    \begin{align}
         \ket{b_0} &= \ket{\phi}\\
        \ket{b_1} &= \frac{\ket{\xi}-\braket{\phi}{\xi}\ket{\phi}}{\sqrt{1-|\braket{\phi}{\xi}|^2}} = \frac{\ket{\xi}-\braket{\phi}{\xi}\ket{\phi}}{\sqrt{1-\cos(\theta)^2}}.
    \end{align}
   With $\braket{\xi}{\phi} = \cos(\theta)e^{i\alpha}$ we can express $\ket{\psi}$ in the basis as
    \begin{align}
        \ket{\psi} = \cos(\eta)e^{-i\alpha}\ket{b_0} + \sin(\eta)\ket{b_1},
    \end{align}
     With that we find
    \begin{align}
         &\braket{\xi}{\psi}=\cos(\eta)\cos(\theta) + \sin(\eta)\frac{1-\cos(\theta)^2}{\sqrt{1-\cos(\theta)^2}}\\
         =&\cos(\eta)\cos(\theta)+\sin(\eta)\sin(\theta) = \cos(\theta-\eta). 
    \end{align}
  We can now obtain a lower bound on $\left||\braket{\xi}{\psi}|^2-|\braket{\xi}{\phi}|^2\right|$ as follows
     \begin{align}
         &\left||\braket{\xi}{\psi}|^2-|\braket{\xi}{\phi}|^2\right|=\left|\cos(\theta-\eta)^2-\cos(\theta)^2\right|\\
         = &\left|\frac{e^{2i\theta-2i\eta}+2+e^{2i\eta-2i\theta}}{4}-\frac{e^{2i\theta}+2+e^{-2i\theta}}{4}\right|\\
    =&\left|\frac{e^{2i\theta-2i\eta}+e^{2i\eta-2i\theta}-e^{2i\theta}-e^{-2i\theta}}{4}\right|\\
 =&\left|\frac{\left(e^{-2i\eta}-1\right)e^{2i\theta}+\left(e^{2i\eta}-1\right)e^{-2i\theta}}{4}\right| \\
     =&\left|\frac{\left(e^{-i\eta}-e^{i\eta}\right)e^{2i\theta-i\eta}+\left(e^{i\eta}-e^{-i\eta}\right)e^{-2i\theta+i\eta}}{4}\right|\\
     =&\left|\frac{\left(e^{-i\eta}-e^{i\eta}\right)\left(e^{i(2\theta-\eta)}-e^{-i(2\theta+\eta)}\right)}{4}\right|\\
    =&\left|\sin(\eta)\sin(2\theta-\eta)\right| = |\sin(\eta)||\sin(2\theta-\eta)|
     \end{align}
    From $0\leq \cos(\theta)<\nicefrac{1}{\sqrt{2}}$ it follows that $\nicefrac{\pi}{2}>\theta>\nicefrac{\pi}{4}$ and we choose $\eta$ such that $\sin(\eta)=\frac{1}{2\sqrt{N}}$ and $\sin(2\theta-\eta)\geq\sin(\eta)$. Using \ref{eq:noBasisMeas}
     \begin{align}
         &\hdia{RF(\V{A}_{\psi})(\bar a)}{RF(\V{B}_{\psi})\left(\M\circ\S^{QSV}_{\phi, K}\right)}+ \hdia{RF(\pi)(\bar a)}{RF(\sharp_S)\left(\M \circ\S^{QSV}_{\phi, K}\right)}\nnn&\geq \frac{\sin^2(\eta)}{2}\left(1-\sqrt{1-\cos(\eta)^{2N}}\right)\nnn
         &=\frac{\sin^2(\eta)}{2}\left(1-\sqrt{1-(1-\sin(\eta)^2)^{N}}\right)\geq \frac{1}{16N}.
     \end{align}
    \subsubsection{Unital channels}
    A channel is unital when it preserves the identity. That is, for a space $M_a(\Cc)$ a unital channel $\Lambda$ maps $M_a(\Cc)$ to $M_a(\Cc)$ and it holds that
      \begin{align}
          \Lambda(\mathbbm{1}_{M_a(\Cc)}) = \mathbbm{1}_{M_a(\Cc)}.
      \end{align}
    Again, we consider a setting as described in section~\ref{section:general_protocol}  with $p(i, i)=0$ where we post-compose the protocol with a unital channel $\Lambda$. We need to find a lower bound on the distinguishing advantage after applying the unital channel. Using the bound in Theorem \ref{thm:general_multi} for mixed states, because $\Lambda(\phi)$ might be mixed even when $\phi$ is pure, we find for arbitrary $\rho$   \begin{align}\label{eq:post_compose_unital_bound}
        \hdia{RF(\V{A}_{\rho})(\bar a)}{RF(\V{B}_{\rho})(\Lambda\circ \S^{QSV}_{\phi, K})}+ \hdia{RF(\pi)(\bar a)}{RF(\sharp_S)\left(\Lambda\circ\S^{QSV}_{\phi, K}\right)}\geq \nnn\frac{\|\Lambda(\rho)-\Lambda(\phi)\|_1}{4}\left(1-\sqrt{1-(1-\|\rho-\phi\|_1)^{N}}\right).
     \end{align}
    We can freely choose $\rho$, and set the following:
    \begin{align}
        \rho = \alpha\phi + (1-\alpha)\frac{\mathbbm{1}-\phi}{d-1}.
     \end{align}
    And we  find
    \begin{align}
        &\left\|\Lambda(\phi) - \Lambda(\rho)\right\|_1
        = \left\|\phi'-\left(\alpha \phi'+(1-\alpha)\frac{\mathbbm{1}-\phi'}{d-1}\right)\right\|_1\\
        = &\left\|(1-\alpha)\left(\phi'-\frac{\mathbbm{1}-\phi'}{d-1}\right)\right\|_1
        =(1-\alpha)\left\|\frac{(d-1)\phi'-\mathbbm{1}+\phi'}{d-1}\right\|_1\\
        = &(1-\alpha)\frac{d}{d-1}\left\|\phi'-\frac{\mathbbm{1}}{d}\right\|_1.
    \end{align}
    With $\beta = 1-\alpha$, we obtain
    \begin{align}
       \frac{1}{2}\|\Lambda(\phi)-\Lambda(\rho)\|_1= \frac{\beta d}{d-1}\frac{1}{2}\|\Lambda(\phi)-\nicefrac{\mathbbm{1}}{d}\|_1.
    \end{align}
    Similarly, the trace distance of the inputs is then
    \begin{align}
       \frac{1}{2}\|\rho-\phi\|_1 = \frac{1}{2}\left\|\phi-\left(\alpha\phi+(1-\alpha)\frac{\mathbbm{1}-\phi}{d-1}\right)\right\|_1 = \frac{\beta d}{(d-1)}\frac{1}{2}\|\phi-\nicefrac{\mathbbm{1}}{d}\|_1.
    \end{align}
    We define
    \begin{align}
        \omega &= \frac{d}{d-1}\frac{1}{2}\|\phi-\mathbbm{1}/d\|_1,\\
        \omega' &= \frac{d}{d-1}\frac{1}{2}\|\Lambda(\phi)-\mathbbm{1}/d\|_1.
    \end{align}
We can then rewrite \ref{eq:post_compose_unital_bound} as follows
    \begin{align}
        &\hdia{RF(\V{A}_{\rho})(\bar a)}{RF(\V{B}_{\rho})(\Lambda\circ \S^{QSV}_{\phi, K})}+ \hdia{RF(\pi)(\bar a)}{RF(\sharp_S)\left(\Lambda\circ\S^{QSV}_{\phi, K}\right)}\nnn &\geq \frac{\omega'\beta}{2}\left(1-\sqrt{1-(1-2\omega\beta)^{N}}\right).
    \end{align}
    We set $\beta = \frac{1}{2\omega N}$, and we can assume that $N\geq \frac{1}{2\omega N}$, since $\phi$ does not depend on the number of rounds. Using $(1-\frac{\beta}{k})^{k}\geq 1-\beta$ for $k\in\mathbb{N}$ and $|\beta|\leq k$ we find:
    \begin{align}
        &\hdia{RF(\V{A}_{\rho})(\bar a)}{RF(\V{B}_{\rho})(\Lambda\circ \S^{QSV}_{\phi, K})}+ \hdia{RF(\pi)(\bar a)}{RF(\sharp_S)\left(\Lambda\circ\S^{QSV}_{\phi, K}\right)}\nnn
        \geq &\frac{\omega'}{2(2\omega N)}\left(1-\sqrt{1-\left(1-\frac{1}{N}\right)^{N}}\right)
        \geq \frac{\omega'}{4\omega N}.
    \end{align}